\documentclass[leqno]{amsart}
\usepackage{amsfonts,amssymb,amsmath,amsgen,amsthm}
\usepackage{fancyhdr}
\usepackage{color}
\usepackage[svgnames]{xcolor}
\xdefinecolor{color}{named}{Crimson}

\theoremstyle{plain}
\newtheorem{theorem}{Theorem}[section]
\newtheorem{definition}[theorem]{Definition}

\newtheorem{lemma}[theorem]{Lemma}
\newtheorem{corollary}[theorem]{Corollary}
\newtheorem{proposition}[theorem]{Proposition}

\theoremstyle{remark}
\newtheorem{remark}[theorem]{Remark}
\newtheorem*{notation}{Notation}

\font\tenms=msbm10
\font\sevenms=msbm7
\font\fivems=msbm5
\newfam\bbfam \textfont\bbfam=\tenms \scriptfont\bbfam=\sevenms
\scriptscriptfont\bbfam=\fivems

\let\vf=\varphi

\def\eps{\varepsilon}

\def\cdotv{\raise 2pt\hbox{,}}

\def\C{{\mathbf C}}
\def\R{{\mathbf R}}
\def\N{{\mathbf N}}
\def\Sch{{\mathcal S}}

\def\virgp{\raise 2pt\hbox{,}}

\def\({\left(}
\def\){\right)}
\def\<{\left\langle}
\def\>{\right\rangle}

\def\d{{\partial}}

\numberwithin{equation}{section}

\title{Nonlinear Propagation of coherent states through avoided energy level crossing.}
\author[L. ~Hari]{Lysianne Hari}
\address[L. Hari]{University of Cergy-Pontoise\\
                  UMR CNRS 8088\\
                  F-95000 Cergy-Pontoise.}
\email{Lysianne.Hari@u-cergy.fr}

\begin{document}
\begin{abstract}
We study the propagation of wave packets for a one-dimensional system of two coupled Schr\"odinger equations with 
a cubic nonlinearity, in the semi-classical limit. Couplings are induced by the nonlinearity and by the potential, 
whose eigenvalues present an ``avoided crossing'': at one given point, the gap between them reduces 
as the semi-classical parameter becomes smaller. 
For data which are coherent states polarized along an eigenvector of the potential, 
we prove that when the wave function propagates through the avoided crossing point, 
there are transitions between the eigenspaces at leading order. 
We analyze the nonlinear effects, which are noticeable away from the crossing point, but see that 
in a small time interval around this point, the nonlinearity's role is negligible at leading order, 
and the transition probabilities can be computed with the linear Landau-Zener formula.
\end{abstract}
\maketitle
\let\thefootnote\relax\footnotetext{This work was supported by ERC Grant DISPEQ.}
\section{Introduction}
In the framework of the Born-Oppenheimer approximation, systems of linear time-dependent Schr\"odinger equations 
have been studied throughout years in order to understand molecular dynamics. 
The notion of adiabaticity and questions about energy level crossing arose. 
In fact, when one considers systems where the electronic energy levels are assumed to be well isolated from each other, 
one can prove that there is an \textit{adiabatic decoupling}. 
The validity of such appoximations has been analyzed, in various settings (see for instance \cite{MartinezSordoni} and 
\cite{SpohnTeufel} and the references given there); however the approximation breaks down in the presence of eigenvalue crossing, 
leading to numerous questions about those situations.
\\
Thus several types of eigenvalue crossing phenomena have been analyzed, 
since they can imply transitions between electronic energy levels.
One of these situations, where the adiabatic approximation breaks down is when one has an \textit{Avoided crossing}, 
as studied in \cite{HJ98}, \cite{HJ99}, \cite{HagClass}, and \cite{Rousse04}.
\begin{definition}
\label{avoided}
 Let $d,n\in \N^{*}$, and $\Omega \subset \R^d$ an open subset of $\R^d$.
We suppose that $V_\delta(x)$ is a family of $n\times n$ symmetric and smooth matrices on $\Omega$ and 
$\delta \in [0,\delta_0)$, for a fixed $\delta_0>0$. 
 We suppose that $V_\delta(x)$ has two eigenvalues $\lambda_\delta^{\mathcal{A}}(x),\lambda_\delta^{\mathcal{B}}(x)$ 
 such that for all $x \in \Omega$, $\lambda_\delta^{\mathcal{A}}(x)\neq \lambda_\delta^{\mathcal{B}}(x)$, for $\delta >0$.
 We consider 
 $\Gamma$ given by
 $$\Gamma = \left\lbrace x \;| \;\lambda_0^{\mathcal{A}}(x)=\lambda_0^{\mathcal{B}}(x)\right\rbrace,$$ and assume that 
 $\Gamma$ is a single point or a non-empty connected proper submanifold of $\Omega$.
\\
 Then, we say that $V_\delta(x)$ has an \textbf{Avoided crossing} on $\Gamma$.
\end{definition}
In this case, two energy levels come close to one another, without crossing and a solution with a data polarized along 
one given mode is not polarized along this mode anymore, when it has propagated through the avoided crossing point. 
In the linear case, it is possible to compute the transition probabilities, 
thanks to the well known \textit{Landau-Zener formula} (\cite{Landau}, \cite{Zener}), 
which was mathematically proved in \cite{Hag91}, and in \cite{J94} with less restrictions. 
Propagation of specific coherent states through avoided crossing has been studied in \cite{HJ98} 
(see also \cite{HagClass} for a classification) and in \cite{Hag94} for ``exact'' crossing, 
where some issues about regularity are added in the study.
For more general data and crossing, the transition is also noticeable when one looks at 
the Wigner transform of the wave function, whose description is performed studying semi-classical measures and is useful 
to understand how the Wigner transform concentrates on trajectories passing through the crossing region 
(see for instance  \cite{FG02} for semi-classical measures describing the Wigner transform in an explicit case and \cite{FG03} 
for more general data). 
The crossing phenomena have been analyzed in the linear case and the reader can find various results on different aspects, 
including \cite{CdV1}, \cite{CdV2} for a classification of results for more general equations. 
For some numerical simulations, that are used in other fields such as Quantum Chemistry, 
we refer the reader to see for instance \cite{FL08} and \cite{LT}.
\\[1mm]
\indent
The aim of this paper is to analyze the nonlinear twin of the situation presented in \cite{HJ98} 
in a simple and explicit case, in order to understand the nonlinear effects combined with the crossing phenomenon: 
we will study an avoided crossing phenomenon, which occurs at one point, for a system of two nonlinear time-dependent 
Sch\"odinger equations with an initial coherent state, in dimension one, with a cubic nonlinearity. 
\\
This problem arises from the description of nonadiabatic transitions when one studies properties of 
binary mixtures of Bose-Einstein Condensates (see \cite{Cornell1}, \cite{Cornell2}, \cite{Cornell3}). 
The nonlinearity induces a coupling between each mode and if the systems do present crossing phenomena, 
we want to understand how effects from both couplings can interact. Some cases with a potential without eigenvalue crossing 
have been studied, in order to analyze the nonlinear effects: in \cite{Adiab} and \cite{Hari1}, for initial coherent states, 
adiabatic theorems and validity of approximations of the wave function at leading order are proved 
for cubic nonlinearities, provided there is a gap assumption. 
Note also that other point of views, such as stationnary problems, are also discussed, for instance in \cite{AftaBlancNier}, 
using Ginzburg-Landau Energies in the context of Binary mixtures of Bose-Einstein Condensates. 
Some results in different frameworks, such as quantum systems with periodic potentials, 
and transitions between Bloch bands, can be found for instance in \cite{NLZper},\cite{ErrNLZper}. 
The result presented here extends to more general nonlinearities, of the form $(a|\psi^\eps_1|^2+b|\psi^\eps_2|^2)\psi^\eps$ 
such as whose studied in \cite{Cornell1}, \cite{Cornell2}, \cite{Cornell3} and references there. 

\subsection{Framework}
We consider the semi-classical limit $\eps\to 0$ for the nonlinear
Schr\"odinger equation 
\begin{equation}
\label{NLS0}
     \left\lbrace \begin{array}{l} 
               i\eps\d_t \psi^\eps+ \dfrac{\eps^2}{2}\d_x^2 \psi^\eps - V_\delta(x)\psi^\eps = \kappa \eps^{3/2}|\psi^\eps|^2\psi^\eps; \\
               \psi^\eps(-T,.)\in \Sch(\R),
\end{array}\right.
\end{equation}
where $\psi^\eps(t,x)=\left(\psi^\eps_1(t,x);\psi^\eps_2(t,x)\right)$, $(t,x)\in \R \times \R$, 
$\kappa \in \R$ is a small coefficient, and 
the quantity $|\psi^\eps|^2$ denotes the square of the
Hermitian norm in $\C^2$ of the vector $\psi^\eps$. 
The initial data $\psi^\eps(-T,.)$ is a \textit{coherent state} or \textit{wave packet}, which concentrates 
at some given point in the phase space.
We consider the following potential $$V_\delta(x)=\begin{pmatrix}
               x & \delta \\
               \delta & -x
              \end{pmatrix}, \textrm{ where } \delta >0,$$
for $\delta = c\sqrt{\eps}$ for a nonnegative constant $c$.
The eigenvalues are given by $$\lambda_\delta^\pm (x) = \pm \sqrt{x^2+\delta^2}, \; \textrm{and}\;  \lambda_0^+(0)=\lambda_0^-(0)= 0.$$ 
One can see that the gap size between the two eigenvalues is at least $\delta>0$, and that it is minimal when $x=0$ 
and so that an avoided crossing phenomenon occurs at $x=0$. 
\\
\noindent
The eigenvectors associated with the eigenvalues of $V_\delta$ are
$$\chi_\delta^\pm(x) = \left( \begin{array}{c}
                                                  \Theta_1^\pm(x)\\
                                                  \Theta_2^\pm(x)
                                                 \end{array}\right),$$                                                  
with                                          
\begin{equation}
\label{theta12}
\left\lbrace 
\begin{array}{cc}
 \Theta_1^+(x) =-\Theta_2^-(x)&= \dfrac{\delta}{\left(2\sqrt{x^2+\delta^2}(\sqrt{x^2+\delta^2}-x)\right)^{1/2}},\\
 \\
 \Theta_2^+(x) =\Theta_1^-(x) &= \dfrac{\sqrt{x^2+\delta^2}-x}{\left(2\sqrt{x^2+\delta^2}(\sqrt{x^2+\delta^2}-x)\right)^{1/2}}.
\end{array}\right.
\end{equation}
One can see that since $\delta >0$, the eigenvectors are $\mathcal{C}^\infty$.
In the following, we will use estimates on their derivatives: writing $\Theta_j^\pm$ as $\Theta_j^\pm(x)=f_j^\pm(x,\delta),$ where $f_j^\pm$ are homogeneous 
of degree 0 and we obtain 
\begin{equation}
 \label{eigenvectors}
 \left|\d_x^\alpha \chi_\delta^\pm (x)\right| \lesssim \delta^{-\alpha}, \quad \forall \alpha \in \N.
\end{equation}
It is also possible to obtain finer bounds in some cases, see Section \ref{eigen} for a deeper discussion about these eigenvectors.
\\[2mm]
Let us first comment on the value of the parameter $\delta$.
In the linear case, studied in \cite{HJ98}, the authors work with a gap size of the order $\sqrt{\eps}$ since 
there can be important transitions between each mode for $\delta$ of this size. The same critical value appears in the study of 
Dirac type equations, studied in \cite{FermanianDirac}. For higher power of $\eps$, 
one can prove that the adiabatic decoupling is still valid. We choose to study the nonlinear propagation in this setting too. 
For asymptotics in linear cases, with other values of $\delta$, which can be independent of $\eps$ we refer the reader to~\cite{Rousse04}.
\\[2mm]
Another notion of criticality appears for the exponent of the nonlinearity. We recall that if we write the nonlinearity 
 $\kappa \eps^{\alpha}|\psi^\eps|^{2\sigma}\psi^\eps$ and denote by $d$ the space dimension, we say that:
\begin{itemize}
 \item the nonlinearity is $L^2-$\textit{subcritical} if $\sigma < 2/d$, $L^2-$\textit{supercritical} otherwise.
 \item for $d\geq 3$, the nonlinearity is $H^1-$\textit{subcritical} if $\sigma < 2/(d-2)$.
\end{itemize}
Here, the nonlinearity is $L^2-$subcritical which is a good point to prove global existence of the solution 
for fixed $\eps$ more easily, and to deal with other technical issues that will be developped later.
\\
The coefficient $\kappa$ can be either negative or nonnegative, but has to be small:
$\exists C>0$ independent of $\eps, \delta$, such that $|\kappa|\leq 1/C.$ For convenience, we will use $|\kappa| \leq 1$ and will 
then make a restriction and take it smaller in the analysis, when it will be needed.
\\
It is also worth pointing out that the nonlinearity is critical for semi-classical wave packets, in the case without crossing: 
if we write the nonlinearity $\kappa \eps^{\alpha}|\psi^\eps|^2\psi$, and introduce $\alpha_c=1+d\sigma/2=3/2$, 
for data which are wave packets, we have
\begin{itemize}
 \item if $\alpha>\alpha_c$, one can linearize the equation at leading order, 
 since the nonlinearity's weight (or in an other point of view, the size of the initial data) is not big enough to have an effect. 
 One can build an approximation which is a linear coherent state (at leading order), in the case of adiabatic regimes.
 \item if $\alpha=\alpha_c$, for this critical situation, the effects of the nonlinearity cannot be neglected: 
 in the absence of crossing points, one can still approach the wave function by a coherent state at leading order, but it will 
 get some nonlinear effects.
\end{itemize}
We refer the reader to \cite{CF-p} for a deeper discussion about this critical exponent, in a scalar case, 
and \cite{Adiab}, \cite{Hari1} for matrix-valued cases. Similar discussions are made for Hartree equations in \cite{APPP11}, 
\cite{CaoCarles}.
\\
 For fixed $\eps$, and for $\delta >0$, since the potential is at most quadratic, in view of \cite{Ca-p}, 
 one can prove global existence and uniqueness of the solution $\psi^\eps$ to \eqref{NLS0}, for any data in 
 $\Sch(\R)$ and for any $\kappa \in \R$. 
\\
Since our aim is to understand the competition between couplings induced by the nonlinearity and those induced by the potential, 
if there is some, we choose a critical power of $\eps$ in front of the nonlinearity.
\subsection{Classical trajectories and actions}
We introduce the following quantities.
\\
\textit{Classical trajectories.} Let $(x^\pm(t), \xi^\pm(t))$ be the solution of the following system: 
  \begin{equation}
   \label{trajectories}
   \left\lbrace \begin{array}{l}
  \dot{x}^\pm(t)=\xi^\pm(t)\\
  \dot{\xi}^\pm(t)=-\partial_x \lambda_\delta^\pm(x^\pm(t))
 \end{array}\right. \qquad \textrm{with} \quad \left\lbrace \begin{array}{l}
  x^\pm(0)=x^\pm_0 \\
  \xi^\pm(0)=\xi^\pm_0 
 \end{array} \right.
  \end{equation}
\begin{remark}
 These trajectories admit a limit when $\delta$ tends to zero:
 $$\left\lbrace \begin{array}{l}
                 \dot{x}^\pm(t)=\xi^\pm(t)\\
                 \dot{\xi}^\pm(t) = \dfrac{\mp x(t)}{|x(t)|},
                \end{array}  \right.$$
Note that these limit trajectories are well-defined (see for instance \cite{FG03}).
\end{remark}
\begin{remark}
\label{rmk:traj}
Since $\delta>0$, the eigenvalues are smooth and so \eqref{trajectories} has a unique, global, $\delta-$dependent and smooth solution. 
Besides, for any $T>0$, using \cite{CF-p} one can write
$$\exists C>0, \; \exists \delta_0>0,  \; \forall t \in [-T,0],\; \forall \delta \in ]0,\delta_0], \; |\xi^\pm(t)| + |x^\pm(t) |\leq C.$$
\end{remark}

\noindent
\textit{Classical action.} $$S^\pm(t) = \int_0^t \dfrac{|\xi^\pm(s)|^2}{2}- \lambda_\delta^\pm(x^\pm(s)) ds.$$
 
\subsection{Initial time and Data} 
We consider the ``$+$'' classical trajectories introduced by \eqref{trajectories}, defined by choosing 
$$x_0=0 \quad ; \quad \xi_0>0, $$
as initial data for them. 
The point of minimal gap is reached at time $t=0$ for $x=0$.
\\
We chose an initial data $\psi^\eps(-T,x)$  which is a localized wave packet, 
polarized along the eigenvector $\chi_\delta^+$:
\begin{equation}
 \label{data}
 \psi^\eps(-T,x) = \eps^{-1/4}a\left(\dfrac{x-x^+(-T)}{\sqrt{\eps}}\right)
e^{\frac{iS^+(-T)}{\eps}+ \frac{i\xi^+(-T).(x-x^+(-T))}{\eps}}\chi^+_\delta(x),
\end{equation}
where the terms $x^+(-T), \xi^+(-T), S^+(-T)$ are the values of the quantities, introduced by 
\eqref{trajectories}, 
at time $t=-T$, and $a \in \Sch(\R)$.
\\
Moreover, we restrain to $T$ such that for $t\in[-T,0]$, the classical trajectory $x^+(t)$ grows to zero.

\subsection{Main result}
We will split the analysis of the propagation of the coherent state into three parts, 
depending on the closeness of the crossing region; we will have to consider different time intervals. 
Each one will lead to a different regime, and the approximations will have to be matched at the 
border of the time intervals. 
\\
It is necessary to consider different regimes because of nonadiabatic transitions: in fact, 
the wave function cannot remain localized in the mode ``$+$'' at leading order beyond a specific time, 
which is a small power of $\eps$. Thus, in order to deal with energy level transitions, 
we have to build different approximations.
\\[1mm]
We first introduce some functions.
\\[1mm]
In the adiabatic region, where $x^+(t) \ll 0$, we consider the function $u_\delta = u_\delta(t,y)$, solution to
\begin{equation}
 \label{profile}
 i\d_t u_\delta+ \dfrac{1}{2} \d_y^2 u_\delta - \dfrac{1}{2} \lambda_\delta^{+(2)}(x^+(t))y^2 u_\delta = 
 \kappa |u_\delta|^2 u_\delta \quad ; \quad
 u_\delta(-T,y)=a(y),
\end{equation}
where $a\in \Sch(\R)$, and 
$\lambda_\delta^{+(2)}(x) = \delta^2 \left(x^2+\delta^2\right)^{-3/2}.$ Existence and properties of $u_\delta$ are discussed in 
Section \ref{Proof thm:profile}.
\\[1mm]
In the crossing region, we introduce $f$, a vector-valued function, solution to
\begin{equation}
 \label{f}
 i\d_s f - \left(  \begin{array}{cc}
            y+s\xi_0 & c \\
            c & -(y+s\xi_0)
           \end{array}
\right)f = 0 ;
\end{equation}
with data
\begin{equation}
 \label{dataf}
  f(-c_0\eps^{-\gamma},y) = u_\delta(-c_0\eps^{1/2-\gamma},y)\left(\begin{array}{c}
                                                                  0 \\1
                                                                 \end{array}
 \right)e^{\frac{i}{\eps}\phi^\eps(y)}
\end{equation}
where $s,y\in \R$, $\gamma \in ]0,1/6[$, and $\phi^\eps(y)$ is a real-valued phase function (see \eqref{phasephi} 
for an explicit formula). Note that the system \eqref{f} presents transitions between each mode.
\\[2mm]
We can now state the main theorem of the paper, 
which gives a valid approximation of the exact solution $\psi^\eps$ at leading order, in the limit $\eps \rightarrow 0$:
\begin{theorem}[Main Theorem]
 \label{main}
 We consider $\psi^\eps (t,x)$ the exact solution to the Cauchy problem~\eqref{NLS0} with data \eqref{data}, and 
 $\delta = c\sqrt{\eps}$ for some $c>0$.
 Then, if $c_0>0$ is independent of $\eps$, and if $\gamma \in ]0,1/6[$:
 \\[2mm]
 $(1)$ For $-T\leq t \leq -c_0 \eps^{\frac{1}{2}-\gamma}$, in the limit $\eps \rightarrow 0$ we have
 $$\psi^\eps(t,x) = \eps^{-1/4}u_\delta\left(t,\dfrac{x-x^+(t)}{\sqrt{\eps}}\right)
 e^{\frac{iS^+(t)}{\eps}+\frac{i\xi^+(t).(x-x^+(t))}{\eps}}\chi_\delta^+(x)+\mathcal{O}(\eps^{\gamma}), \;\textrm{in }L^2,$$
where $u_\delta$ is the profile, solution to \eqref{profile}.
\\[2mm]
 $(2)$ For $-c_0\eps^{\frac{1}{2}-\gamma}\leq t \leq c_0\eps^{\frac{1}{2}-\gamma}$, in the limit $\eps \rightarrow 0$ we have
 $$\psi^\eps(t,x) = \eps^{-1/4}f\left(\dfrac{t}{\sqrt{\eps}},\dfrac{x-t\xi_0}{\sqrt{\eps}}\right)
 e^{\frac{i\xi_0^2t}{2\eps}+\frac{i\xi_0.(x-t\xi_0)}{\eps}} + \mathcal{O}(\eps^{\gamma/2}), \;\textrm{in }L^2$$
 where $f$ is the solution to \eqref{f} and with data \eqref{dataf}.
\end{theorem}
\begin{remark}
 For both approximations given by points $(1)$ and $(2)$ from the previous theorem, to deal with the nonlinearity, 
 we also need estimates in some weighted Sobolev spaces of type $H^1$ (see Theorems \ref{thm:IncomingOuter} and \ref{thm:Inner}). 
\end{remark}

\noindent
For time $t \in [-T,c_0\eps^{1/2-\gamma}]$, the exact solution $\psi^\eps$ can be approached, 
at leading order, by a coherent state polarized along the same eigenvector as the initial data. 
The nonlinear effect is noticeable thanks to the profile $u_\delta$, 
whose equation is a nonlinear Schr\"odinger equation \eqref{profile}. 
\\
The transition between energy levels occurs on the time interval 
$$[-t^\eps,t^\eps]=[-c_0\eps^{1/2-\gamma},c_0\eps^{1/2-\gamma}]$$ but the nonlinearity, 
does not affect the phenomenon.
\begin{corollary}[Transition between the modes]
 \label{cor:main}
 We consider $$\psi^\eps (t,x)= \psi^\eps_+(t,x)\chi_\delta^+(x) + \psi^\eps_-(t,x)\chi_\delta^-(x)$$ 
 the exact solution to the Cauchy problem~\eqref{NLS0} with data \eqref{data}, and 
 $\delta = c\sqrt{\eps}$ for some $c>0$.
 Consider $c_0>0$, $\gamma \in ]0,1/6[$, and $t^\eps = c_0\eps^{1/2-\gamma}$, when $\eps \rightarrow 0$
 \begin{align*}
    \forall t \in [-T,-t^\eps], \quad \|\psi^\eps_+(t)\|_{L^2}^2 & = \|a\|_{L^2}^2+o(1)\\
  \|\psi^\eps_-(t)\|_{L^2}^2 &= o(1).
 \end{align*}
and for $p=e^{-\frac{\pi c^2}{\xi_0}}$
\begin{align*}
 \|\psi^\eps_+(t^\eps)\|_{L^2}^2 & = (1-p) \|a\|_{L^2}^2+o(1)\\
  \|\psi^\eps_-(t^\eps)\|_{L^2}^2 &= p\|a\|_{L^2}^2 + o(1).
\end{align*}
\end{corollary}
Note that $p$ is the same transition coefficient as in \cite{HJ98} (the Landau-Zener coefficient). 
\\[2mm]
We are not able to describe the wave function for nonnegative times of order $\mathcal{O}(1)$ because the approximation given by
Theorem \ref{main} is not good enough to be propagated. 
In fact, in simplest situations, as in \cite{Adiab} and \cite{Hari1}, one can treat initial data which are a wave packet up to 
a term of size $\mathcal{O}(\eps^{l})$, in $L^2(\R^d) \cap H^1_\eps(\R^d)$ where $l>d/8$, and can prove 
the validity of the approximation.
\\
In our case, the fact that $0<\gamma<1/6$ yields a technical obstruction if we keep the same approach for $t\in [t^\eps,T]$ 
as we did for $t\in [-T,-t^\eps]$.
\subsection{Organization of the paper}
In Section 2, we give a brief exposition without proofs of all results we need in order to prove the main theorem of the paper. 
We then look more closely at the profile $u_\delta$ in Section 3 before proceeding with the study of the approximation far from 
the crossing region in Section 4. The fifth section is devoted to the analysis of the crossing region and gives the proof of 
the validity of the second approximation. In the last section, we restrict our attention to the transition phenomenon.
\section{Sketch of proof of the main theorem}
\subsection{Approximation away from the crossing point}
Our aim is to approach the exact solution $\psi^\eps$ by a function polarized along the eigenvector $\chi_\delta^+(x)$, 
that we are going to build using the profile $u_\delta$, up to a time $-t^\eps$ of order $\eps^{1/2-\gamma}$, for some $\gamma \in ]0,1/6[$.
near the crossing point.
\\[1mm]
We consider the classical trajectories and action associated with $\lambda_\delta^+(x)$. 
In order to prove the validity of the approximation of point $(1)$ in Theorem \eqref{main}, 
we first need to control $u_\delta$, solution to \eqref{profile}. 
By \cite{Adiab}, we have global existence of $u_\delta$ for any $\delta >0$.
\begin{theorem}[Global existence from \cite{Ca-p}]
Let $\delta >0$ and $\kappa \in \R$. For all $a\in \Sch(\R)$, \eqref{profile} has a unique solution 
 $$u_\delta \in \mathcal{C}(\R, L^2(\R))\cap L_{loc}^{8}(\R, L^4(\R)).$$
 Moreover, its $L^2-$norm is conserved:
 $$\| u_\delta(t)\|_{L^2(\R)}= \|a\|_{L^2(\R)}, \quad \forall t\in \R.$$
\end{theorem}
\begin{proof}[Sketch of the proof]
 This result is proved in \cite{Ca-p}: for a fixed $\delta>0$, 
 the potential is at most quadratic. Thus, one can use local in time Strichartz estimates (which are available thanks to 
 \cite{Fujiwara79} and \cite{Fujiwara}) and prove local existence of the solution. 
 And since we are in the $L^2-$subcritical case, the mass conservation implies global existence of the profile for any $\kappa$, 
 in the suitable space.
 \end{proof}
 \begin{proposition}[Control of derivatives and momenta far from $t=0$]
For any $T_0 \in ]0,T[$, there exists $C>0$, such that 
 $$\forall t \in[-T,-T_0], \; \forall \alpha, \beta \in \N, \; \alpha+\beta \leq k, 
 \quad \|y^\alpha \d_y^\beta u_\delta(t) \|_{L^2(\R)}\leq C.$$
 \end{proposition}
\begin{proof}[Sketch of the proof]
The only thing we need to observe is that
$$\left|\lambda^{+ (2)}(x(t))\right| = \dfrac{\delta^2}{\left(x^+(t)^2+\delta^2\right)^{3/2}} 
\leq \dfrac{C\delta^2}{\left(\inf_{[-T,-T_0]} x^+(t)\right)^3}\leq \widetilde{C}\delta^2,$$
for $t\in [-T,-T_0]$, since $|x^+(t)|$ is bounded from below by a positive constant independent of $\eps$ far from $t=0$.
 The control of the derivatives and momenta is then a consequence of \cite{Ca-p}.
\end{proof}
 We actually need the profile on a bounded time interval of the form $[-T,0]$, and it requires additionnal work to prove 
 the uniformity of the bound since this interval contains zero. 
 The following result will be proved in Section \ref{Proof thm:profile}:
\begin{theorem}[Behaviour of derivatives and momenta until $t=0$]
\label{thm:profile}
 Let $\kappa \in \R$, and $a\in \Sch(\R)$. We consider $u_\delta$ the solution to the Cauchy problem \eqref{profile},
 and $\delta = c\sqrt{\eps}$ for some $c>0$. 
 Then, there exists $T_0>0$ such that for all $k\in \N$, the following property is satisfied:
 \\
 There exists $C>0$, such that 
 $$\forall \alpha, \beta \in \N, \; \alpha+\beta \leq k, \quad \|y^\alpha \d_y^\beta u_\delta(t) \|_{L^2(\R)}\leq C, 
 \quad \forall t\in [-T_0,0]$$
\end{theorem}
\noindent
As a consequence, we have a control on the whole interval $[-T,0]$.
\\[2mm]
\noindent
We denote by $\varphi^\eps$, the following function associated with $u_\delta$, $x^+$, $\xi^+$, $S^+$:
\begin{equation}
\label{phi}
 \varphi^\eps(t,x) = \eps^{-1/4}u_\delta \left(t, \dfrac{x-x^+(t)}{\sqrt{\eps}}\right)
 e^{\frac{iS^+(t)}{\eps}+\frac{i\xi^+(t).(x-x^+(t))}{\eps}}.
\end{equation}
Using $|\xi^+(t)|\leq C$, for $t\in [-T,0],$ we deduce the following corollary.
\begin{corollary}
 \label{phiLinfty}
Let $T>0$. Let us consider $a\in \Sch(\R)$, and $u_\delta$ solution to \eqref{profile}. We have for $\varphi^\eps$ defined 
by \eqref{phi}
$$\forall \beta \in \N, \quad \|\eps^\beta \d_x^\beta \varphi^\eps(t) \|_{L^\infty(\R)} \lesssim \eps^{-1/4}, 
\; \forall t \in [-T,0].$$
\end{corollary}
\vspace{2mm}
\noindent
The first point of Theorem \ref{main} is a consequence of the following theorem, that will be proved 
in Section \ref{proof thm:IncomingOuter}.
\begin{theorem}
 \label{thm:IncomingOuter}
Let $T >0$, $\gamma \in [0,1/6]$, $c_0>0$ and $a \in \Sch(\R)$. We consider $\psi^\eps$, 
the exact solution to the Cauchy problem \eqref{NLS0} - \eqref{data}, and $\vf^\eps$, the function given by \eqref{phi}.
Then the difference $$w^\eps(t,x) = \psi^\eps (t,x) - \vf^\eps(t,x)\chi_\delta^+(x),$$
satisfies:
$\exists C>0$, $\exists \eps_0>0$, $\forall \eps \in ]0, \eps_0]$
\begin{equation*} 
\sup_{t \in [-T,-t^\eps]} \| w^\eps (t) \|_{L^2(\R)}+\|\eps\d_x w^\eps(t)\|_{L^2(\R)}\leq C \eps^{\gamma},
\end{equation*}
where $t^\eps = c_0 \eps^{1/2-\gamma}$.
\end{theorem}

\subsection{Through the crossing point: comparison with the linear model}
\label{introInnerRegime}
As in the linear case, studied by the authors of \cite{HJ98}, the classical trajectories and action $x^+$, $\xi^+$ and $S^+$ 
are not relevant to build an approximation when one approaches the avoided crossing point: 
the aim is to show that the exact solution does not remain in the energy level ``$+$'' and so, 
that after time $t=0$, the wave function is not utterly localized around $(x^+(t),\xi^+(t))$ in the phase space, at leading order. 
We are going to prove that it has two components on each mode at time $t^\eps = c_0\eps^{1/2-\gamma}$. 
Thus, on this small time interval, we localize the approximation around an
``averaged'' trajectory which is close to both trajectories ``$+$'' and ``$-$'' near zero. 
Indeed both $\pm$ trajectories satisfy 
$$x^\pm(t) = \xi_0t+\mathcal{O}(t^2) \quad ; \quad \xi^\pm(t)= \xi_0 + \mathcal{O}(t),$$
so we introduce the free trajectory, which is an approximation of the previous ones:
 $$ \widetilde{x}(t) = \xi_0t \quad ; \quad \widetilde{\xi}(t)=\xi_0>0.$$
\\[1mm]
In the nonlinear case, the main difficulty that could arise is the presence of the nonlinear contribution at leading order 
on this time interval. We will see that with a critical nonlinearity, at leading order, there won't be any effect and 
the mechanism of transition is guided by the system \eqref{f}, as in the linear case (see Equation (3.91) in \cite{HJ98}).
We introduce the following rescaled variables:
$$\left\lbrace \begin{array}{c}
y = \left(x-t\xi_0\right)/\sqrt{\eps}\\
                s=t/\sqrt{\eps},
                
               \end{array}
   \right. $$
and the rescaled solution $v^\eps$ is given by
\begin{equation}
 \label{rescaled}
 \psi^\eps(t,x)= \eps^{-1/4}v^\eps \left(\dfrac{t}{\sqrt{\eps}}, \dfrac{x-t\xi_0}{\sqrt{\eps}} \right)
e^{\frac{i\xi_0^2t}{2\eps}+\frac{i\xi_0.(x-t\xi_0)}{\eps}}
\end{equation}
with $v^\eps(s,y) \in \C^2$, satisfying the following Schr\"odinger equation
\begin{equation}
 \label{NLS1}
 \left\lbrace \begin{array}{l}
  i\d_s v^\eps + \dfrac{\sqrt{\eps}}{2}\d_x^2 v^\eps - V_{\delta/\sqrt{\eps}}\left(y+s\xi_0\right) v^\eps = 
 \kappa \sqrt{\eps} |v^\eps|^2v^\eps,\\
 v^\eps(-c_0\eps^{-\gamma},y)= v^\eps_{init.}(y)
 \end{array}\right.
\end{equation}
where 
$$V_{\delta/\sqrt{\eps}}\left(y+s\xi_0\right) = \begin{pmatrix}
                                                (y+s\xi_0) & \delta/\sqrt{\eps} \\
                                                \delta/\sqrt{\eps} & -(y+s\xi_0)
                                               \end{pmatrix}.$$
The function $v^\eps_{init.}(y)$ is choosen so that \eqref{rescaled} holds at $t=-c_0\eps^{1/2-\gamma} $ 
and will be given in details in Section \ref{match}.
There exist nonlinear Landau-Zener formulae adapted to nonlinear transition systems; see for instance the following system from 
\cite{CF-LZ} and \cite{BiaoQian}.
$$i\d_t u = \mathcal{H}(\gamma)u,\quad u=(u_1,u_2), $$
with $$ \mathcal{H}(\gamma) = \begin{pmatrix}
                                                                              \gamma(t)+\kappa(|u_2|^2-|u_1|^2)& \delta \\
                                                                              \delta & -(\gamma(t)+\kappa(|u_2|^2-|u_1|^2))
                                                                             \end{pmatrix},$$
where $\delta$ is the coupling constant between the energy levels, $\kappa$ is a parameter for the nonlinear interaction, 
and $\gamma(t)$ is the level separation. The reader can also refer to \cite{KR} or \cite{Sacchetti09} 
for other discussions about the nonlinear versions of the Landau-Zener formula.
However, we will not use these tools in this paper since the nonlinear effects will not be visible in our transition system.
\\
In fact, $v^\eps$ can be approached at leading order by the linear function $f$ solution to \eqref{f}-\eqref{dataf}, on the 
time interval $[-c_0\eps^{1/2-\gamma},c_0\eps^{1/2-\gamma}]$. 
Let us first notice that an explicit form of the solution $f$ is computed in \cite{HJ98}, using \textit{parabolic cylinder functions}, 
whose asymptotics are well known (see \cite{GR80} for details).
We need some results on $f$, how it is ``carried'' on each eigenspace, in order to deduce some information on $v^\eps$. 
Writing $f_{1,2}(s,y) = u_{2,1}(s\sqrt{\xi_0}+y/\sqrt{\xi_0},y)$, we have the following theorem from \cite{FG02}, 
which contains the only results we need on the asymptotics:
\begin{theorem}[Scattering result from Appendix 9 in \cite{FG02}]
 \label{diffuFG}
 We consider the following system
 \begin{equation}
 \label{LZsystem}
  -i\d_s u = \begin{pmatrix}
              s & \eta \\
              \eta & -s 
             \end{pmatrix}u.
 \end{equation}
Then there exist $\left(g_1^{s-},g_2^{s-}\right)$ and $\left(g_1^{s+},g_2^{s+}\right)$ two orthonormal bases of solutions to 
\eqref{LZsystem} such that locally uniformly in $\eta$, the following asymptotics hold:
\begin{equation*}
 \textrm{for }s \rightarrow - \infty, \quad \left\lbrace \begin{array}{cl}
                                                  g_1^{s-}(s,\eta) =& e^{i\Lambda(s,\eta)}  \left( \begin{array}{c}
         1 \\ 0
        \end{array}\right) + o(1)  \\
                                                  g_2^{s-}(s,\eta) =& e^{-i\Lambda(s,\eta)} \left( \begin{array}{c}
         0 \\ 1
        \end{array}\right) + o(1)
                                                 \end{array}
 \right.
\end{equation*}
\begin{equation*}
 \textrm{for }s \rightarrow - \infty, \quad \left\lbrace \begin{array}{cl}
                                                  g_1^{s+}(s,\eta) =& e^{i\Lambda(s,\eta)} \left( \begin{array}{c}
         1 \\ 0
        \end{array}\right)+ o(1)  \\
                                                  g_2^{s+}(s,\eta) =& e^{-i\Lambda(s,\eta)} \left( \begin{array}{c}
         0 \\ 1
        \end{array}\right)+ o(1)
                                                 \end{array}
 \right.
\end{equation*}
 where $$\Lambda(s,\eta) =  \dfrac{s^2}{2}+\dfrac{\eta^2}{2}\log|s|,$$
 and with the following transition rule from components $(\alpha_1,\alpha_2)$ in $\left(g_1^{s-},g_2^{s-}\right)$ to
 $(\beta_1,\beta_2)$ in $\left(g_1^{s+},g_2^{s+}\right)$:
 \begin{equation*}
  \left(\begin{array}{c}
  \beta_1(\eta) \\
  \beta_2(\eta)
  \end{array}\right) = \begin{pmatrix}
                        a(\eta) & -\overline{b}(\eta) \\
                        b(\eta) & a(\eta)
                       \end{pmatrix}  \left(\begin{array}{c}
  \alpha_1(\eta) \\
  \alpha_2(\eta)
  \end{array}\right)
 \end{equation*}
where
$$a(\eta) = e^{-\pi\eta^2/2}, \quad b(\eta) = \dfrac{2i}{\sqrt{\pi}\eta}2^{-i\eta^2/2}e^{-\pi\eta^2/4}
\Gamma\left(1+\dfrac{i\eta^2}{2}\right)\sinh\left(\dfrac{\pi\eta^2}{2}\right),$$
 and $|a(\eta)|^2+|b(\eta)|^2=1.$
\end{theorem}
As a consequence, for $f$ we obtain
\begin{corollary}
 \label{cor:f}
 Let $\gamma \in ]0,1/6[$. 
 We consider $f=(f_1,f_2)$ the solution to \eqref{f}-\eqref{dataf} on $[-c_0\eps^{-\gamma},c_0\eps^{-\gamma}] $.
 Then, 
 $$\|f_1(-c_0\eps^{-\gamma})\|_{L^2}^2 = 0 \quad \textrm{and} 
 \quad \|f_2(-c_0\eps^{-\gamma})\|_{L^2}^2 = \|a\|_{L^2}^2+o(1), $$
and 
\begin{align*}
 \|f_1(c_0\eps^{-\gamma})\|_{L^2}^2 &= (1-e^{-\frac{\pi c^2}{\xi_0}})\|a\|_{L^2}^2 +\;o(1)\\
 \|f_2(c_0\eps^{-\gamma})\|_{L^2}^2 &= e^{-\frac{\pi c^2}{\xi_0}}\|a\|_{L^2}^2+\;o(1).
\end{align*}
\end{corollary}
\noindent
Here again, the behaviour of $f$ and its derivatives is important.
\begin{lemma}
 \label{lemma:f}
 Let $\gamma \in ]0,1/6[$ and $f$ be the solution to \eqref{f} with an initial data given by \eqref{dataf}. 
 Then if $s \in [-c_0\eps^{-\gamma}, c_0\eps^{-\gamma}]$, 
 we have:
 $$\forall k \in \N,\;  \exists C_k>0, \quad \|\d_y^k f(s)\|_{L^2}\leq C_k \; \eps^{-k\gamma}.$$
\end{lemma}
\noindent
Then, we state the following theorem which implies the second point of Theorem \ref{main} and is proved in Section \ref{inner}.
\begin{theorem}
 \label{thm:Inner}
 Let $a \in \Sch (\R)$, $\gamma \in ]0,1/6[$ and $c_0>0$. We consider $v^\eps$ the solution to \eqref{NLS1}, 
 on the time interval $I_\eps = [-c_0\eps^{-\gamma},c_0\eps^{-\gamma}]$, 
 and $f$ the solution to the linear ODE  \eqref{f}-\eqref{dataf} on the same time interval. 
 Then, the function $$r^\eps(s,y)=v^\eps(s,y)-f(s,y)$$ satisfies: 
 $\exists C_1, C_2>0$, $\exists \kappa_0>0$, $\forall \kappa \in ]0,\kappa_0]$, 
 $\exists \eps_0>0$, $\forall \eps \in ]0, \eps_0]$, 
 
 $$ \sup_{s \in I_\eps}\left\|r^\eps(s) \right\|_{L^2} \leq C_1\eps^{\gamma/2} \quad; \quad  
 \sup_{s \in I_\eps} \left\|\sqrt{\eps}\d_yr^\eps(s) \right\|_{L^2} = C_2\eps^{\gamma}.$$
\end{theorem}
Let us mention that Corollary \ref{cor:main} is a consequence of Corollary \ref{cor:f} and Theorem \ref{thm:Inner} as we shall see in 
Section \ref{conclusion}.

\section{Properties of the profile $u_\delta$ in the adiabatic region}
\label{Proof thm:profile}
In this section, we prove Theorem \ref{thm:profile} and study the profile, solution to the Cauchy problem \eqref{profile} 
when $t=0$. 
There are two difficulties. The first one is linked with the small size of the gap $\delta$ since the function 
$\lambda_\delta^{(2)} (x(t))= \mathcal{O}(\delta^{-1})$ when $t$ tends to zero. 
The second one is that we have no Strichartz estimates 
for the operator $$-\dfrac{1}{2}\d_y^2 + \dfrac{1}{2}\lambda_\delta^{(2)} (x(t))y^2.$$
\\
In order to avoid dealing with the $\delta-$dependent potential, we use a Lens transform that allows to drop the potential 
and to use ``free Strichartz estimates''.

\subsection{Preliminary results}
We introduce the tools and some results that we need in our proofs.
\\
In many computations in this paper, the function $\lambda_\delta^{+(2)}(x(t))$ has to be controlled.
We will use the following lemma.
\begin{lemma}[From Proposition 2.2 of \cite{HJ98}]
\label{key}
Let $T>0$. Then, there exists $\sigma>0$, such that
 $$\exists \delta_0>0, \; \forall \delta \in ]0,\delta_0], \quad 
 \int_{0}^t \left|\lambda_\delta^{(2)}(x(s))\right| ds \leq \sigma(1+T), \quad \forall \;0<t\leq T.$$
\end{lemma}
From now on, we shall take $\delta \in ]0,\delta_0]$.
\begin{proof}[Proof from \cite{HJ98}]
Asymptotic estimates when $t\rightarrow 0$ and $\delta \rightarrow 0$:
 \\
 We write $x(t) = x(0)+t.\dot{x}(0) + \mathcal{O}(t^2) = t.\xi_0 + \mathcal{O}(t^2)$. Then 
$x(t)^2=t^2\xi_0^2+\mathcal{O}(t^3)$. This gives for $\lambda_\delta^{(2)}(x(t))$:
\begin{align}
 \notag \lambda_\delta^{(2)}(x(t)) & = \dfrac{\delta^2}{\left(x(t)^2+\delta^2\right)^{3/2}}\\
 \notag & = \dfrac{\delta^2}{\left(\xi_0^2t^2+\delta^2\right)^{3/2}} \left(1+\mathcal{O}\left(\dfrac{t^3}{t^2\xi_0^2+\delta^2}\right) \right)\\
 \label{asymptotic}
 & = \dfrac{\delta^2}{\left(\xi_0^2t^2+\delta^2\right)^{3/2}}+ \mathcal{O}(1).
\end{align}
Then we use \eqref{asymptotic} to write
\begin{align*}
 \int_{0}^t \left|\lambda_\delta^{(2)}(x(s))\right| ds & \leq \sigma \int_{0}^t \dfrac{\delta^2}{(s^2+\delta^2)^{3/2}}+1 \; ds \\
 & \leq \sigma \int_{0}^{T} \dfrac{\delta^2}{(s^2+\delta^2)^{3/2}}+1 \; ds\\
 & \leq \sigma \left[ \dfrac{s}{(s^2+\delta^2)^{1/2}} +s\right]^{T}_0 \leq \sigma(1+T).
\end{align*}
\end{proof}
\noindent
We now introduce $\mu_\delta, \nu_\delta$, solutions to 
 \begin{equation}
 \label{eq:munu}
     \begin{cases}
                \ddot{\mu_\delta}+\lambda_\delta^{(2)}(x(t))\mu_\delta=0; \quad  \mu_\delta(0)=0; & \dot{\mu_\delta}(0)=1,\\
                \ddot{\nu_\delta}+\lambda_\delta^{(2)}(x(t))\nu_\delta=0; \quad  \nu_\delta(0)=1; & \dot{\nu_\delta}(0)=0.
     \end{cases}
\end{equation}
Let us notice that for fixed $\delta>0$, there exists a unique couple of solutions $(\mu_\delta,\nu_\delta)$ satisfying \eqref{eq:munu} 
on some maximal interval of existence $[0,T_\delta]$.
\begin{proposition}
  \label{munu}
 We consider $\mu_\delta,\nu_\delta$ solutions to \eqref{eq:munu}. 
 There exist $T_0>0$ such that $T_\delta \leq T_0,$ independently of $\delta$,
 and $C>0$ such that 
 $$\forall \delta \in ]0,\delta_0], \forall t \in [0,T_0], \; |\mu_\delta(t)|+|\nu_\delta(t)|+ |\dot{\nu}_\delta(t)|+ \left|\dfrac{1}{\nu_\delta(t)}\right|
 \leq C.$$
 Moreover, $\displaystyle t\mapsto\dfrac{\mu_\delta(t)}{\nu_\delta(t)}$ is an increasing function on $[0,T_0]$.
 \end{proposition}
In the rest of this paper, in order to simplify the notations, we will not write the dependence on $\delta$ of these functions.
Lemma \ref{key} is a crucial tool to obtain the preceding proposition.
 \begin{proof}
Let us first study $\nu$ and $\dot{\nu}$.
\\
For $ 0\leq t \leq T_0$, we write:
$$\dfrac{d}{dt} \left( \begin{array}{c}
                        \nu \\
                        \dot{\nu}
                       \end{array}
\right) = \begin{pmatrix}
           0 & 1 \\
           - \lambda_\delta^{(2)}(x(t)) & 0
          \end{pmatrix} \left( \begin{array}{c}
                        \nu \\
                        \dot{\nu}
                       \end{array}\right) \quad ; \quad  \left( \begin{array}{c}
                        \nu(0) \\
                        \dot{\nu}(0)
                       \end{array}\right) =  \left( \begin{array}{c}
                        1 \\
                        0
                       \end{array}\right).
 $$
We deduce

$$\left( \begin{array}{c}
                        \nu(t) \\
                        \dot{\nu}(t)
                       \end{array}
\right)=\left( \begin{array}{c}
                        1 \\
                        0
                       \end{array}\right) + \int_{0}^t \left( \begin{array}{c}
                        \dot{\nu}(s) \\
                        - \lambda_\delta^{(2)}(x(s))\nu(s)
                       \end{array}\right) ds. $$
\noindent                       
We introduce the following notation:
$$|A|_\infty = \sup_{s \in [0,T_0]} |A(s)|.$$
\noindent
We have
 \begin{align*}
  \left|\int_{0}^t \dot{\nu}(s) ds \right| \leq |\dot{\nu}|_\infty |t|, \\
 \left|\int_{0}^t - \lambda_\delta^{(2)}(x(s)) \nu(s) ds \right| &\leq |\nu|_\infty \int_{0}^{T_0} |\lambda_\delta^{(2)}(x(s))| ds 
 \leq \sigma(1+T_0) |\nu|_\infty. 
 \end{align*}
It gives
\begin{align*}
|\nu|_\infty & \leq 1+ |t| \; |\dot{\nu}|_\infty\\
|\dot{\nu}|_\infty & \leq \sigma(1+T_0) |\nu|_\infty, 
\end{align*}
and then
\begin{align*}
|\nu|_\infty & \leq 1+ \sigma(1+T_0) \;|t| \; |\nu|_\infty \\
 |\dot{\nu}|_\infty & \leq \sigma(1+T_0) |\nu|_\infty. 
\end{align*}
We choose $T_0$ small enough at the beginning, such that for all $ 0\leq t \leq T_0$, we have
\begin{equation*}
 \sigma(1+T_0)\; |t| \leq \dfrac{1}{2}, \quad \textrm{where } T_0 \leq \dfrac{\sqrt{1+2/\sigma}-1}{2},
\end{equation*}
then
\begin{equation*}
|\nu|_\infty  \leq C_1 \quad ; \quad |\dot{\nu}|_\infty  \leq C_2,
\end{equation*}
where $C_1, C_2$ are independent of $\delta$.
\\[2mm]
We now study $\left(\nu(t)\right)^{-1}$. Since $\nu(0)=1$ we have
$$\left|\nu(t)-1\right|_\infty \leq |t| \left|\dot{\nu}(t) \right|_\infty. $$
We deduce $1- C_2 |t| \leq \left|\nu(t)\right|_\infty$, which gives for $t$ small enough 
$\left|\nu(t)\right|_\infty \leq \dfrac{1}{2}$ and so $\left(\nu(t)\right)^{-1} \leq 2$.
Using the same arguments, we finally compute for $\mu$:
\begin{align*}
|\mu|_\infty & \leq |t|+\sigma|t|(1+T_0)|\mu|_\infty \\
 |\dot{\mu}|_\infty & \leq 1+\sigma(1+T_0)|\mu|_\infty. 
\end{align*}
And so, choosing $T_0$ small enough, as we have already done, we find
\begin{equation*}
 |\mu|_\infty \leq 2 |t| \leq 2T_0\leq C_4,
\end{equation*}
and taking $C= \max \left\lbrace C_1,C_2,C_3,C_4 \right\rbrace $, the desired estimates are proved.
\\[2mm]
It is then easy to see that $t \mapsto \dfrac{\mu(t)}{\nu(t)}$ is an increasing function on our time interval: 
in fact $$\dfrac{d}{dt} \left(\dot{\mu}\nu-\dot{\nu}\mu\right) = \ddot{\mu}\nu-\ddot{\nu}\mu=0, $$
and so $\dot{\mu}\nu-\dot{\nu}\mu = 1$ for all $t \in [0,T_0]$. We deduce that
$$\left(\dfrac{\mu}{\nu}\right)' = \dfrac{\dot{\mu}\nu-\dot{\nu}\mu}{\nu^2} = \dfrac{1}{\nu^2},$$ is nonnegative on $[0,T_0]$.
\end{proof}
\subsection{Lens transform and Free Strichartz estimates}
In order to remove the potential, we now introduce the Lens transform, from \cite{Ca-p}, which relies upon the functions $\mu,\nu$:
\begin{proposition}
 \label{Lens}
Let $a\in \Sch(\R)$ and $T_0>0$, $\mu,\nu$ as in Proposition \ref{munu}. We set $$s=\dfrac{\mu(t)}{\nu(t)}, \quad t\in [0,T_0],$$
and $$\mathcal{H}(s)=\mathcal{H}_\delta(s)=\nu\left(\left(\dfrac{\mu}{\nu}\right)^{-1}(s)\right)\kappa, \quad 
s\in \left[0,\dfrac{\mu(T_0)}{\nu(T_0)}\right]. $$
We consider $v$ the solution to:
 \begin{equation}
  \label{v}
  i\d_s v+ \dfrac{1}{2} \d_x^2 v = \mathcal{H}(s) |v|^2 v \quad ; \quad
 v(0,x)=u_\delta(-T_0,x).
 \end{equation}
Set
 $$u_\delta(t-T_0,x)= \dfrac{1}{\sqrt{\nu(t)}} v\left( \dfrac{\mu(t)}{\nu(t)}, \dfrac{x}{\nu(t)}\right) 
e^{\frac{i\dot{\nu}(t)}{\nu(t)}.\frac{x^2}{2}}, \quad \forall t\in[0,T_0],$$ then $u_\delta$ is the solution to \eqref{profile}:
\begin{equation*}
 i\d_t u+ \dfrac{1}{2} \d_x^2 u - \dfrac{1}{2} \lambda_\delta^{(2)}(x(t))x^2u = \kappa|u|^2 u \quad ; \quad
 u(-T,x)=a(x).
\end{equation*}
\end{proposition}
\begin{remark}
 In \cite{Ca-p}, the statement is given in a more general case, for higher dimensions, with potentials
 of the form $\left\langle \Omega(t)x|x\right\rangle $, with $\Omega \in \mathcal{C}(\R,\R)$ and isotropic. 
 More general nonlinearities, with time-dependent coefficients, are also allowed under specific conditions.
\end{remark}
\begin{remark}
\label{H}
 Using Proposition \ref{munu}, it is easy to see that $\dfrac{\mu(T_0)}{\nu(T_0)}$ is bounded independently of $\delta$ and so is 
 $\mathcal{H}(s)$ on $\left[0,\dfrac{\mu(T_0)}{\nu(T_0)}\right]$.
\end{remark}
 The conservation of mass of the profile $u$ and Proposition \ref{munu} allow us to infer existence of the solution $v$ 
 for fixed $\delta$, on the bounded time interval $\left[0,\dfrac{\mu(T_0)}{\nu(T_0)}\right]$. 
 The next step is to study the derivatives and momenta of $v$, solution to \eqref{v}. 
 We introduce \textit{free Strichartz estimates}, from \cite{GV95}, \cite{Seg76} and \cite{Strichartz}.
\begin{definition}[Admissible pairs]
\label{def:adm}
 A pair $(p,q)$ is {\bf admissible} if $2\leq q
  \leq \infty$  
  and 
$$\frac{2}{p}= \left( \frac{1}{2}-\frac{1}{q}\right).$$
\end{definition}
We now introduce the following notation to state Strichartz estimates.
\begin{notation}
 For $1\leq p \leq +\infty$, we denote by $p'$, the dual exponent:
$$\dfrac{1}{p}+\dfrac{1}{p'}=1.$$
\end{notation}
\begin{theorem}
 \label{Strichartz} 
Let $(p,q)$, $(p_1, q_1)$, $(p_2, q_2)$ be admissible pairs and let $I$ be a finite time interval.
Let us introduce $$u(t)=e^{-i\frac{t}{2}\d_x^2} u_0 \qquad \textrm{and} \qquad
 v(t)= \int_{I\cap \left\lbrace s \leq t \right\rbrace}  e^{i \frac{t-\tau}{2}\d_x^2}f(\tau) d\tau.$$
\noindent
$-$ There exists $C=C(q)$ such that for all $u_0 \in L^2(\R)$, we have for all $s \in I$
\begin{equation}
\label{Strichartz1}
\|u \|_{L^p (I, L^q(\R))} \leq C \|u(s)\|_{L^2(\R)}=C \|u_0\|_{L^2(\R)}.
\end{equation}
$-$ There exists $C = C(q_1, q_2)$ such that for all 
\\
$f \in L^{p'_2} (I, L^{q'_2}(\R))$ we have
\begin{equation}
\label{Strichartz2}
\|v \|_{L^{p_1} (I, L^{q_1}(\R))} \leq C \|f \|_{L^{p'_2} (I, L^{q'_2}(\R))}.
\end{equation}
\end{theorem}

\subsection{Control of $v$}
Using the previous results, we can study the derivatives and momenta of $v$, given by \eqref{v}.
\begin{proposition}
 \label{growthV}
 Let $a \in \Sch(\R)$ and $T_0>0$ the fixed time given by Proposition~\ref{munu}. 
 Let $v$ be the solution to the Cauchy problem \eqref{v}, with 
 $$\mathcal{H}(s) = \nu \left(\left(\dfrac{\mu}{\nu}\right)^{-1} (s)\right)\kappa.$$
Then, for all $k\in \N$, there exists $C_k>0$ such that for all admissible pairs $(p,q)$ and for all $\delta \in ]0,\delta_0]$
\begin{equation}
\label{ControlV}
\forall \alpha, \beta \in \N, \; \alpha+\beta = k, \quad 
\|x^\alpha \d_x^\beta v\|_{L^p\left(\left[0,\frac{\mu(T_0)}{\nu(T_0)}\right],L^q\right)} \leq C_k.
\end{equation}
\end{proposition}
\begin{proof}
 We proceed by induction on $k$.
 \\
 \textbf{Case $k=0$:}
 Thanks to the conservation of the $L^2-$norm of $u$, and Proposition \ref{munu} we deduce that \eqref{ControlV} is true for the pair 
 $(\infty,2)$. We know prove \eqref{ControlV} for an other admissible pair. 
 We consider an interval $I=[s, s+\tau]$, where $s\in \left[0,\dfrac{\mu(T_0)}{\nu(T_0)}\right[ $ and $\tau>0$.
 Using Strichartz estimates given by Theorem \ref{Strichartz}, we write
 \begin{align*}
 ||v||_{L^8(I,L^4)} & \leq \widetilde{C} \left[||v(s)||_{L^2}+ \left|\left|\mathcal{H}(s)|v|^2v\right|\right|_{L^{8/7}(I,L^{4/3})}\right] \\
 & \leq {C} \left[||a||_{L^2}+ ||v||^2_{L^{8/3}(I,L^4)} ||v||_{L^8(I,L^4)}\right]\\
 & \leq {C} \left[||a||_{L^2}+ \sqrt{\tau}||v||^2_{L^8(I,L^4)} ||v||_{L^8(I,L^4)}\right]\\
 & \leq {C} \left[||a||_{L^2}+ \sqrt{\tau}||v||^3_{L^8(I,L^4)}\right],
\end{align*}
where we have used Remark \ref{H}, H\"older inequalities and that $||v(s)||_{L^2}$ is bounded independently of $\delta$.
\\
\noindent
We use the following Bootstrap argument to prove that $||v||_{L^8(I,L^4)}$ is bounded:
\\
We suppose $||v||_{L^8(I,L^4)}\leq M$, for $M>0$. Then, the previous estimate gives
\begin{equation*}
 ||v||_{L^8(I,L^4)}\leq \dfrac{M}{2},
\end{equation*}
if \begin{align*}
C\|a\|_{L^2} \leq \dfrac{M}{4} \quad &\textrm{and} \quad C\sqrt{\tau} \|v\|^3_{L^8(I,L^4)} \leq \dfrac{M}{4}  \\
C\|a\|_{L^2} \leq \dfrac{M}{4} \quad &\textrm{and} \quad \tau \leq \left(\dfrac{1}{4CM^2}\right)^2.
   \end{align*}
We define:
$$M= 4C\|a\|_{L^2}+1 \quad ; \quad \tau_0 =\dfrac{1}{(4CM^2)^2}.$$
Finally, we have for $I=[s,s+\tau]$, with $\tau \leq \tau_0$
\begin{equation*}
  ||v||_{L^8(I,L^4)} \leq C \left(\|a\|_{L^2}+{M} \right).
\end{equation*}
Since we only have a finite number of intervals of size $\tau_0$ in the whole considered interval, we can deduce
$$||v||_{L^8\left(\left[0,\frac{\mu(T_0)}{\nu(T_0)}\right],L^4\right)} \leq C, $$
where $C$ is independent of $\delta$. The other pairs are easily obtained, using Strichartz estimates and the previous bound.
\\[2mm] 
\textbf{Case $k \geq 1$: }We now consider some $k$; and assume that 
$$\forall \alpha, \beta \in \N, \; \alpha+\beta = k-1, \quad 
\|x^\alpha \d_x^\beta v\|_{L^p\left(\left[0,\frac{\mu(T_0)}{\nu(T_0)}\right],L^q\right)} \leq C_{k-1} $$ 
The equations for the $k-$derivative and momentum are:
\begin{align*}
 i\d_s(\d_x^k v) + \dfrac{1}{2} (\d_x^k v) &= \mathcal{H}(s) \d_x^k(|v|^2v)\\
 i\d_s(x^k v) + \dfrac{1}{2} (x^k v) &= \mathcal{H}(s) |v|^2(x^k v)+ \dfrac{1}{2} \left[\d_x^2,x^k \right]v. 
\end{align*}
We recall the following result, which is a consequence of \cite{Triebel}, 
and that will be useful to drop all ``mixed'' norms of derivatives and momenta:
\begin{equation*}
 \sum_{|\alpha|+|\beta| \leq k} \|x^\alpha \d_x^\beta v \|_{L^q} \leq C 
\left[ \|(1+|x|)^{k}v \|_{L^q} + \sum_{|\alpha|\leq k} \|\d_x^\alpha v\|_{L^q} \right].
\end{equation*}
We consider $I=[s,s+\tau]$, $\tau>0$ and thanks to Strichartz estimates and Remark \ref{H}:
\begin{align*}
 \|\d_x^k v\|_{L^\infty(I,L^2)}+ \|\d_x^k v\|_{L^8(I,L^4)} & \lesssim \|\d_x^k v(s)\|_{L^2} +  \|\d_x^k (|v|^2v)\|_{L^{8/7}(I,L^{4/3})}\\
 \|x^k v\|_{L^\infty(I,L^2)}+ \|x^k v\|_{L^8(I,L^4)} & \lesssim \|x^k v(s)\|_{L^2} + \| |v|^2(x^k v)\|_{L^{8/7}(I,L^{4/3})}\\
 & \quad+ \|\left[\d_x^2,x^k \right]v \|_{L^1(I,L^2)}.
\end{align*}
We write \begin{align*}
          \left|\d_x^k (|v|^2v) \right| & \lesssim |v|^2 |\d_x^kv| + \sum_{j\in J} |w_{j_1}| |w_{j_2}| |w_{j_3}|, \quad 
          J \subset \N,\\
          \left[\d_x^2,x^k \right]v & = C_1 x^{k-1} \d_x v + C_2 x^{k-2}v,
         \end{align*}
where $w_{j_l}$ is a derivative of $u$ or $\bar{u}$ of order lower than $k-1$. 
So
\begin{align*}
 \|\d_x^k v\|_{L^\infty(I,L^2)}+ \|\d_x^k v\|_{L^8(I,L^4)} & \lesssim \|\d_x^k v(s)\|_{L^2} 
 + ||v||^2_{L^{8/3}(I,L^4)}\|\d_x^k v\|_{L^8(I,L^4)} \\
 + &\sum_{j\in J} \|w_{j_1}\|_{L^8(I,L^4)} \|w_{j_2}\|_{L^8(I,L^4)} \|w_{j_3}\|_{L^8(I,L^4)}\\
 & \lesssim K \left[ \|\d_x^k v(s)\|_{L^2} + \tau^{1/2}\|\d_x^k v\|_{L^8(I,L^4)} + \tau^{1/2}C_{k-1}  \right], 
\end{align*}
where we have used H\"older, the fact that we are in the one-dimensional case, the estimate obtained for $k=0$
and the induction hypothesis. Then, we choose 
$\tau$ such that $K\tau^{1/2} \leq 1/2$ and we obtain
\begin{equation*}
 \|\d_x^k v\|_{L^\infty(I,L^2)}+ \|\d_x^k v\|_{L^8(I,L^4)} \lesssim \|\d_x^k v(s)\|_{L^2} +C.
\end{equation*}
Then, we analyze the other term, with the same tools:
\begin{align*}
\|x^k v\|_{L^\infty(I,L^2)}+ \|x^k v\|_{L^8(I,L^4)} & \lesssim \|x^k v(s)\|_{L^2} + ||v||^2_{L^{8/3}(I,L^4)}\|x^k v\|_{L^8(I,L^4)} \\
 & \quad + \tau\|x^{k-1} \d_x v \|_{L^\infty(I,L^2)}+\tau\|x^{k-2}v\|_{L^\infty(I,L^2)}
 \\
 & \lesssim \|x^k v(s)\|_{L^2} + \tau^{1/2}\|x^k v\|_{L^8(I,L^4)} +C_{k-2}\\
 & \quad + \tau \left[ \|x^k v \|_{L^\infty(I,L^2)}+\|\d_x^k v \|_{L^\infty(I,L^2)}+C_{\leq k-1}\right]\\
 & \leq K \left[ \|x^k v(s)\|_{L^2} + \tau^{1/2}\|x^k v\|_{L^8(I,L^4)}\right. \\
 & \left.\quad + \tau\|x^k v \|_{L^\infty(I,L^2)}+\tau\|\d_x^k v \|_{L^\infty(I,L^2)} + C \right].
\end{align*}
We choose $\tau$ small: $K\tau^{1/2}\leq 1/2$ and $K\tau \leq 1/2$. 
\\
\noindent
For $\tau$ such that $\tau \leq \inf \left(1/(2K); 1/(4K^2) \right)$, we have
\begin{align*}
\|x^k v\|_{L^\infty(I,L^2)}+ \|x^k v\|_{L^8(I,L^4)} & \lesssim \|x^k v(s)\|_{L^2} +\|\d_x^k v(s)\|_{L^2} +C.
\end{align*}
Since we work on a finite time interval, and $\tau$ is bounded by a constant independent of $s, \delta$, 
we have a finite number of intervals of size $\tau$ in it:
\begin{equation*}
 \|x^k v(s)\|_{L^2} +\|\d_x^k v(s)\|_{L^2} \lesssim 1, \quad \forall s \in \left[0,\dfrac{\mu(T_0)}{\nu(T_0)}\right].
\end{equation*}
Bound for other admissible pairs are easily obtained using Strichartz estimates and the previous estimate.
\end{proof}

\subsection{Control of the profile}
\label{subsec:profile}
We will now prove Theorem \ref{thm:profile}. We recall:
$$u_\delta(t-T_0,x)=\dfrac{1}{\sqrt{\nu(t)}} v\left( \dfrac{\mu(t)}{\nu(t)}, \dfrac{x}{\nu(t)}\right) 
e^{\frac{i\dot{\nu}(t)}{\nu(t)}.\frac{x^2}{2}}.$$
The mass-conservation of the profile gives the result for $k=0$. 
For $k \geq 1$, we still analyze full derivatives and momenta of order $k$ only, 
and not all ``mixed'' norms, thanks to \cite{Triebel}. Using the previous formula, we can compute for all $t\in[0,T_0]$:
\begin{align*}
 \| x^k u_\delta(t-T_0)\|_{L^2}^2 & =\nu(t)^{2k} \left\| y^k v\left(\dfrac{\mu(t)}{\nu(t)} \right)\right\|_{L^2}^2 \leq C \\
 \| \d_x^k u_\delta(t-T_0)\|_{L^2}^2 & \leq \int \dfrac{1}{\nu(t)} \sum_{\alpha\leq k}\left|\nu(t)^{-\alpha}
 \d_y^{\alpha}v\left(\dfrac{\mu(t)}{\nu(t)} \right)\right|^2 \times 
 \left(\d_y^{k-\alpha}e^{\frac{i\dot{\nu}(t)}{\nu(t)}.\frac{x^2}{2}} \right)^2 dx
\end{align*}
We use Faa di Bruno's formula to deal with the derivative of the exponential and notice that 
we obtain a sum of terms of the form $\nu^{-2\gamma}\d_x^\gamma v \times \dfrac{\dot{\nu}^\gamma_1}{\nu^{\gamma_2}}y^{\gamma_3}v$, 
with exponents smaller than $2k$. Using Propositions \ref{munu} and \ref{growthV} we deduce that all these terms are bounded 
by constants independent of $\delta$ and $t$; we finally obtain for $t$ instead of $t-T_0$:
$$\| x^k u(t)\|_{L^2}+ \| \d_x^k u(t)\|_{L^2} \leq C, \quad \forall t \in [-T_0,0]. $$
\\[2mm]
\noindent
From this study of the profile, we easily deduce the result on $\varphi^\eps$ stated in Corollary~\ref{phiLinfty}.
\\[2mm]
The task is now to prove the validity of the approximation given by \eqref{phi}.

\section{Convergence of the approximation in the adiabatic region.}
\label{proof thm:IncomingOuter}
In this section, we prove Theorem \eqref{thm:IncomingOuter}. 
To simplify notation, we will drop the sign ``$+$'' when no confusion can arise. 
\subsection{Strategy of the proof}
We recall the equation satisfied by $\psi^\eps$:
$$\left\lbrace \begin{array}{l} 
               i\eps\d_t \psi^\eps+ \dfrac{\eps^2}{2}\d_x^2 \psi^\eps - V_\delta(x)\psi^\eps = \kappa \eps^{3/2}|\psi^\eps|^2\psi^\eps; \\
               \psi^\eps(-T,x) \; \textrm{given by \eqref{data}},
\end{array}\right.$$
and the definition \eqref{phi} of $\varphi^\eps$,
$$\varphi^\eps(t,x) = \eps^{-1/4}u_\delta \left(t, \dfrac{x-x^+(t)}{\sqrt{\eps}}\right)
 e^{\frac{iS^+(t)}{\eps}+\frac{i\xi^+(t).(x-x^+(t))}{\eps}}. $$
We first notice that $\varphi^\eps$ satisfies:
$$ \left\lbrace \begin{array}{l} 
               i\eps\d_t \varphi^\eps+ \dfrac{\eps^2}{2}\d_x^2 \varphi^\eps - \lambda_\delta(x)\varphi^\eps = 
               \kappa \eps^{3/2}|\varphi^\eps|^2\varphi^\eps - \mathcal{R}^\eps_\delta(t,x)\varphi^\eps; \\
               \varphi^\eps(-T,x) =
\eps^{-1/4}a\left(t, \dfrac{x-x(-T)}{\sqrt{\eps}}\right)e^{\frac{iS(-T)}{\eps}+ \frac{i\xi(-T).(x-x(-T))}{\eps}}
\end{array}\right.$$
where 
\begin{equation}
 \label{reste}
 \mathcal{R}^\eps_\delta(t,x) =  \dfrac{\lambda^{(3)}(g(t,x))}{6}(x-x(t))^3, \quad \textrm{with } 
 g(t,x) = \tau x + (1-\tau)x(t), \; \tau \in ]0,1[.
\end{equation}
We want to study the term $w^\eps$ given by $w^\eps = \psi^\eps - \varphi^\eps \chi_\delta$; its equation is
\begin{equation}
\label{w}
\left\lbrace \begin{array}{l} 
               i\eps\d_t w^\eps+ \dfrac{\eps^2}{2}\d_x^2 w^\eps - V_\delta(x)w^\eps = \eps {NL}^\eps+ \eps {L}^\eps; \\
               w^\eps(-T,x)= 0
\end{array}\right.
\end{equation}
where
\begin{align*}
{NL}^\eps & = \kappa \eps^{1/2} \left(|\psi^\eps|^2\psi^\eps - |\varphi^\eps|^2 \varphi^\eps \chi_\delta \right) \\
{L}^\eps & = \eps^{-1}\mathcal{R}^\eps_\delta(t,x)\varphi^\eps \chi_\delta + \eps \d_x \varphi^\eps. d\chi_\delta 
 + \dfrac{\eps}{2} \varphi^\eps d^2 \chi_\delta.
\end{align*}
Using \eqref{eigenvectors}, it is easy to see that some terms are too big (in $L^2$) and will present an obstacle to prove that 
the remainder is small. Thus, we need a finer analysis of the eigenvectors.
\\
Besides, a rough estimate suggests that the second term in ${L}^\eps$ presents an $\mathcal{O}(1)$ contribution. 
Nevertheless after a careful analysis of the eigenvectors to find finer estimates, presented in the next subsection, 
and using the minimal gap size $\delta$, as it is done in \cite{HJ98} we will be able to study all linear terms and their derivatives 
in $L^2$. The last subsection will be devoted to the nonlinear terms, which will be handled by a bootstrap argument, 
and this will complete the proof of Theorem \ref{thm:IncomingOuter}.
\subsection{About the eigenvectors}
\label{eigen}
We need more refined estimates than \eqref{eigenvectors}. We recall that the eigenvectors of the potential are of the form:
$$\chi_\delta^\pm(x) = \left( \begin{array}{c}
                                                  \Theta_1^\pm(x)\\
                                                  \Theta_2^\pm(x)
                                                 \end{array}\right),$$ 
                                                 with
                                                 
\begin{align*}
 \Theta_1^+(x) &= \dfrac{\delta}{\left(2\sqrt{x^2+\delta^2}(\sqrt{x^2+\delta^2}-x)\right)^{1/2}}= -\Theta_2^-(x) ;\\
 \Theta_2^+(x)& = \dfrac{\sqrt{x^2+\delta^2}-x}{\left(2\sqrt{x^2+\delta^2}(\sqrt{x^2+\delta^2}-x)\right)^{1/2}}=\Theta_1^-(x),
\end{align*}
We can rewrite the coordinates:
\begin{equation*}
 \Theta_1^+(x)=F\left(\dfrac{x}{\delta}\right)=-\Theta_2^-(x) \quad ; 
 \quad \Theta_1^-(x) = G\left(\dfrac{x}{\delta}\right)=\Theta_2^+(x), 
\end{equation*}
with 
\begin{align*}
 F(y) = \cos \left(\dfrac{1}{2}\arctan\left(\dfrac{1}{y}\right)\right) \quad &; \quad 
 G(y) = \sin \left(\dfrac{1}{2}\arctan\left(\dfrac{1}{y}\right)\right), \quad \textrm{for }x \geq 0 \\
 F(y) = -\sin \left(\dfrac{1}{2}\arctan\left(\dfrac{1}{y}\right)\right) \quad &; \quad 
 G(y) = \cos \left(\dfrac{1}{2}\arctan\left(\dfrac{1}{y}\right)\right), \quad \textrm{for }x \leq 0.
\end{align*}
We find for $y\neq0$:
\begin{align*}
 \left[\cos \left(\dfrac{1}{2}\arctan\left(\dfrac{1}{y}\right)\right)\right]'
& = \dfrac{-1}{2} \dfrac{1}{1+y^2} \sin \left(\dfrac{1}{2} \arctan \left(\dfrac{1}{y} \right)\right)\\
 \left[ \sin \left(\dfrac{1}{2}\arctan\left(\dfrac{1}{y}\right)\right)\right]'
 &= \dfrac{1}{2} \dfrac{1}{1+y^2} \cos \left(\dfrac{1}{2} \arctan \left(\dfrac{1}{y} \right)\right)\\
 \\
 \left[\cos \left(\dfrac{1}{2}\arctan\left(\dfrac{1}{y}\right)\right)\right]''
 &=\dfrac{-y}{(1+y^2)^2}\sin \left(\dfrac{1}{2} \arctan \left(\dfrac{1}{y} \right)\right) \\
 &\qquad \qquad \qquad - \dfrac{1}{4} \dfrac{1}{(1+y^2)^2}\cos \left(\dfrac{1}{2} \arctan \left(\dfrac{1}{y} \right)\right)\\
 \\
  \left[\sin \left(\dfrac{1}{2}\arctan\left(\dfrac{1}{y}\right)\right)\right]''
  &= \dfrac{y}{(1+y^2)^2}\cos \left(\dfrac{1}{2} \arctan \left(\dfrac{1}{y} \right)\right) \\
 &\qquad \qquad \qquad - \dfrac{1}{4} \dfrac{1}{(1+y^2)^2}\sin \left(\dfrac{1}{2} \arctan \left(\dfrac{1}{y} \right)\right)
\end{align*}
\\[2mm]
So we deduce
\begin{align}
 \label{eigenvectors1}
 |d\chi_\delta^\pm(x)| &\lesssim \dfrac{\delta}{x^2+\delta^2}\lesssim \dfrac{1}{\delta}\dfrac{1}{\left(\dfrac{x}{\delta}\right)^2+1}\\
 \label{eigenvectors2}
 |d^2\chi_\delta^\pm(x)| &\lesssim \dfrac{\delta}{\left(x^2+\delta^2\right)^{3/2}} 
 \lesssim \dfrac{1}{\delta^2}\dfrac{1}{\left(\left(\dfrac{x}{\delta}\right)^2+1\right)^{3/2}}.
\end{align}

\subsection{Analysis of the linear terms}
We recall the form of $L^\eps$ which stands for the linear term of \eqref{w}:
$$L^\eps = \eps^{-1}\mathcal{R}^\eps_\delta\varphi^\eps \chi_\delta + \eps \d_x \varphi^\eps. d\chi_\delta 
 + \dfrac{\eps}{2} \varphi^\eps d^2 \chi_\delta.$$
 In this subsection, we will prove the following lemma:
 \begin{lemma}
  \label{Leps}
  Let $a \in \Sch(\R)$ and $\Lambda>0$ be a fixed constant.  
  Then,  there exists a constant $C_L>0$ independent of $\eps$, such that for all $t \in [-T,-\Lambda \sqrt{\eps}]$, we have
  $$\|L^\eps(t) \|_{L^2} + \|\eps \d_x L^\eps(t)  \|_{L^2} \leq C_L \dfrac{\sqrt{\eps}}{|t|} $$
 \end{lemma}
\noindent
In order to prove this lemma, we will analyze each term and its derivative in $L^2$.
\\[2mm]
\textbf{$\bullet$ Analysis of $\dfrac{\eps}{2} \varphi^\eps d^2 \chi_\delta$}:
\\[1mm]
We see that $\left\|\dfrac{\eps}{2} \; \varphi^\eps. d^2\chi_\delta(t)\right\|_{L^2}  \lesssim \eps |d^2\chi_\delta|_\infty 
 \| \varphi^\eps(t)\|_{L^2} \lesssim \dfrac{\eps}{\delta^2},$
         \\
where we have used \eqref{eigenvectors}. This rough estimate is not good enough: it presents an $\mathcal{O}(1)$ contribution 
for $\delta = c \sqrt{\eps}$. So, using \eqref{eigenvectors2}, we find for $-T\leq t \leq -\Lambda \sqrt{\eps}$:
\begin{align*}
\|\eps \; \varphi^\eps. d^2\chi_\delta(s)\|_{L^1([-T;t],L^2)} & \lesssim
 \int_{-T}^{t} \left[\int_{x\in \R}\eps^2\eps^{-1/2}|d^2 \chi_\delta(x)|^2 |u(s,...)|^2 dx \right]^{1/2} ds \\
& \lesssim \int_{-T}^{t} \left[\int_{|x|\geq \theta|s|} \ldots + 
\int_{|x|\leq \theta |s|} \ldots\right]^{1/2}ds\\
&\lesssim \int_{-T}^{t} \left[I^{1/2}+II^{1/2}\right]ds.
\end{align*}
\vspace{1mm}
For $I$, since $|x| \geq \theta|s|$, we write $\left|\dfrac{x}{\delta}\right|^{-1} \leq \dfrac{\delta}{\theta|s|},$
and obtain
\begin{align*}
 I &\lesssim  \int_{|x|\geq \theta |s|} \eps^2 \eps^{-1/2} \dfrac{1}{\delta^4}\times \dfrac{\delta^6}{|s|^6}|u(s,...)|^2 dx\\
 & \lesssim \dfrac{\eps^2\delta^2}{|s|^6},
\end{align*}
where we have used mass-conservation of the profile. For $II$, we write
\begin{equation*}
 II \lesssim \int_{|x|\leq \theta |s|} \eps^2 \eps^{-1/2}\dfrac{1}{\delta^4} 
 \dfrac{1}{\left(1+\left|\dfrac{x}{\delta} \right|^2 \right)^3} \times 
 \dfrac{\sqrt{\eps}^{2M}}{|x-x(s)|^{2M}}\dfrac{|x-x(s)|^{2M}}{\sqrt{\eps}^{2M}} |u(s,...)|^2 dx \\
\end{equation*}
where $M$ is an integer, $M>1$. 
\\
\noindent
We notice that since $|x|\leq \theta|s|$, we have $|x-x(s)|\geq \xi_0 |s| - \theta |s|$. 
For $\theta$ small enough, we can write $|x-x(s)|\geq \xi_0 |s|/2,$ and so
\begin{equation*}
 II \lesssim \int_{y\in \R} \eps^2 \dfrac{1}{\delta^4} \dfrac{\sqrt{\eps}^{2M}}{|s|^{2M}}\times y^{2M}|u(s,y)|^2 dy \lesssim \dfrac{\eps^2\sqrt{\eps}^{2M}}{\delta^4 |s|^{2M}},
\end{equation*}
where we have used that all momenta of the profile are bounded by constants independent of $\eps$ and $t$. Finally, we compute
\begin{align}
\notag
 \|\eps \; \varphi^\eps. d^2\chi_\delta(s)\|_{L^1([-T;t],L^2)} & \lesssim
 \int_{-\infty}^{t} \dfrac{\eps\delta}{|s|^3} + \dfrac{\eps\sqrt{\eps}^M}{\delta^2 |s|^M} ds \\
 \notag
 & \lesssim \dfrac{\eps\delta}{t^2} + \dfrac{\eps \sqrt{\eps}^M}{\delta^2} \dfrac{1}{t^{M-1}} \\
 \label{lin1}
 & \lesssim \dfrac{\eps\sqrt{\eps}}{t^2}
\end{align}
with $\delta = c \sqrt{\eps}$, and where we have only kept the worst contribution. In fact, thanks to Theorem \ref{thm:profile}, 
all momenta of the profile are bounded on the time interval we consider. So we can choose $M$ as big as we want.  
\\[2mm]
We estimate the derivative in the same fashion: 
\begin{multline*}
 \eps \d_x \left(\eps \; \varphi^\eps. d^2\chi_\delta\right) = \eps^2 \varphi^\eps.d^3\chi_\delta + \\
\eps^{-1/4}e^{\frac{iS(t)}{\eps}+\frac{i\xi(t).(x-x(t))}{\eps}}
\left(i\eps\xi(t)u(\ldots)+\eps\sqrt{\eps}u(\ldots)\right)d^2\chi_\delta.
\end{multline*}
It is sufficient to use \eqref{eigenvectors} to deal with the following terms:
\begin{align*}
 \|\eps^2 \varphi^\eps.d^3\chi_\delta\|_{L^1([-T;t],L^2)}&\lesssim \dfrac{\eps^2}{\delta^3} \lesssim \eps^{1/2}\\
 \|\eps^{-1/4}\eps\sqrt{\eps}u(\ldots)d^2\chi_\delta\|_{L^1([-T;t],L^2)}& \lesssim \dfrac{\eps \sqrt{\eps}}{\delta^2}
 \lesssim \eps^{1/2},
\end{align*}
and we notice that since $|\xi(t)|$ is bounded (Remark \ref{rmk:traj}), we have thanks to \eqref{lin1}
\begin{align*}
 \|\eps^{-1/4}i\eps\xi(s)u(s,\dfrac{x-x(s)}{\sqrt{\eps}})d^2\chi_\delta\|_{L^1([-T;t],L^2)} 
 &\simeq \|\eps \; \varphi^\eps. d^2\chi_\delta(s)\|_{L^1([-T;t],L^2)} \\
 &\lesssim \dfrac{\eps\sqrt{\eps}}{t^2}.
\end{align*}
So, we obtain
\begin{equation}
\label{dlin1}
 \|\eps \d_x \left(\eps \; \varphi^\eps. d^2\chi_\delta\right)\|_{L^1([-T;t],L^2)} \lesssim 
 \sqrt{\eps}+ \dfrac{\eps\sqrt{\eps}}{t^2}.
\end{equation}
\\[3mm]
\textbf{$\bullet$ Analysis of $\eps \d_x \varphi^\eps. d\chi_\delta$}
\\[1mm]
We have $\; \eps \d_x \varphi^\eps. d\chi_\delta = \eps^{-1/4}e^{...}\left[\sqrt{\eps}\d_x u + i\xi(t)u\right] d\chi_\delta,\;$
and we study successively each part. We use \eqref{eigenvectors1} and compute
\\[1mm]
 $\|\eps^{-1/4}\sqrt{\eps}\d_x u.d\chi_\delta(s) \|_{L^1([-T,t];L^2)} $
 \begin{align*}
 & \lesssim \int_{-T}^{t}
 \left[ \int_{x\in \R} \eps \eps^{-1/2} |\d_xu(...)|^2 \dfrac{1}{\delta^2}
 \dfrac{1}{\left(1+\left|\dfrac{x}{\delta}\right|^2\right)^2}dx\right]^{1/2}ds\\
 & \lesssim \int_{-T}^{t} \left( \int_{|x|\geq \theta|s|} \ldots + \int_{|x|\leq \theta|s|} \ldots\right)^{1/2}ds \\
 & \lesssim \int_{-T}^{t} I^{1/2}+II^{1/2}.
\end{align*}
For $I$, since $|x|\geq \theta|s|$, we have 
\begin{equation*}
 I  \lesssim \int_{|x|\geq \theta|s|} \eps \eps^{-1/2} |\d_xu(...)|^2 \dfrac{1}{\delta^2}\dfrac{\delta^4}{|s|^4}dx \lesssim \dfrac{\eps\delta^2}{|s|^4},
\end{equation*}
where we have used that the $H^1-$norm of the profile is bounded. Integrating in time, we obtain 
\begin{equation*}
 \int_{-T}^{t}I^{1/2}ds \lesssim \dfrac{\eps}{|t|}
\end{equation*}
We then write:
\begin{equation*}
II \lesssim \int_{|x|\leq \theta |s|} \eps \eps^{-1/2}\dfrac{1}{\delta^2} 
 \dfrac{1}{\left(1+\left|\dfrac{x}{\delta} \right|^2\right)^2} \times 
 \dfrac{\sqrt{\eps}^{2M}}{|x-x(s)|^{2M}}\dfrac{|x-x(s)|^{2M}}{\sqrt{\eps}^{2M}} |\d_xu(s,...)|^2 dx,
\end{equation*}
and arguing as in the previous analysis, using the control of all momenta of the profile, and integrating in time, we find
\begin{equation*}
 \int_{-T}^{t}II^{1/2}dt \lesssim \dfrac{\sqrt{\eps}^M}{|t|^{M-1}},
\end{equation*}
where $M>1$ can be choosen as big as we want. So, keeping only the worst contribution, we obtain:
\begin{equation*}
 \|\eps^{-1/4}\sqrt{\eps}\d_x u.d\chi_\delta \|_{L^1([-T,t];L^2)} \lesssim \dfrac{\eps}{|t|}.
\end{equation*}
\\[1mm]
We now need to estimate the second contribution of $\eps \d_x \varphi^\eps. d\chi_\delta$, we write 
$$\left| i \eps^{-1/4}\xi(t) u(\ldots).d\chi_\delta\right| \lesssim \left|\dfrac{\delta}{\delta^2+x^2}u(\ldots)\right|. $$
So: $\| i \eps^{-1/4}\xi(t) u(\ldots).d\chi_\delta\|_{L^1([-T,t];L^2)}$
\begin{align*}
 &\lesssim 
 \int_{-T}^{t}\left[ \int_{x\in \R} \eps^{-1/2}\dfrac{\delta^2}{\left(x^2+\delta^2\right)^2}|u(s,\ldots)|^2 dx\right]^{1/2}ds\\
 & \lesssim \int_{-T}^{t}
 \left[ \int_{y\in \R} \dfrac{\delta^2}{\left(\left(x(s)+\sqrt{\eps}y\right)^2+\delta^2\right)^2}|u(s,y)|^2 dy\right]^{1/2}ds\\
 &\lesssim  \int_{-T}^{t} \left(\int_{|y|\geq \theta \frac{|s|}{\sqrt{\eps}}}\ldots\right)^{1/2}
 +\left(\int_{|y|\leq \theta \frac{|s|}{\sqrt{\eps}}}\ldots\right)^{1/2} ds \\
 &\lesssim \int_{-T}^{t} A^{1/2}+B^{1/2},
\end{align*}
with $c>0$. For $A$, we write $|y|\geq \theta \frac{|s|}{\sqrt{\eps}}$ and so $|y|^{-1} \leq \dfrac{\sqrt{\eps}}{\theta|s|}$ and then, 
for an integer $M>1,$ we have
\begin{align*}
 A & \lesssim \int_{|y|\geq \theta \frac{|s|}{\sqrt{\eps}}} \dfrac{\delta^2}{\left(\left(x(s)+\sqrt{\eps}y\right)^2+\delta^2\right)^2}
 \; \dfrac{1}{y^{2M}} \times y^{2M}|u(s,y)|^2 dy \\
 & \lesssim \dfrac{1}{\delta^2} \dfrac{\sqrt{\eps}^{2M}}{|s|^{2M}}\|y^{M}u(s)\|_{L^2}^2,
\end{align*}
and so
\begin{equation*}
 A^{1/2} \lesssim \dfrac{\sqrt{\eps}^{M-1}}{|s|^{M}},
\end{equation*}
where $M>1$ is an integer. This gives
\begin{equation*}
 \int_{-T}^{t} A^{1/2} ds \lesssim \dfrac{\sqrt{\eps}^{M-1}}{|t|^{M-1}}.
\end{equation*}
For the other part, $B$, we first need to find a lower bound for the term $(x(s)+\sqrt{\eps}y)^2.$ We write
\begin{align*}
 \left|x(s) + \sqrt{\eps}y \right| & \geq |x(s)| - \sqrt{\eps} |y| \\
 & \geq \dfrac{|x(s)|}{2} + \dfrac{|x(s)|}{2} - \sqrt{\eps}|y| \\
 & \geq \dfrac{|x(s)|}{2} + \dfrac{\xi_0 |s|}{4} - \theta|s|,
\end{align*}
since $|y| \leq \theta|s|/\sqrt{\eps}$. For a small $\theta>0$, we have 
\begin{equation*}
 \left|x(s) + \sqrt{\eps}y \right| \geq \dfrac{|x(s)|}{2},
\end{equation*}
and so 
\begin{align*}
 \int_{-T}^{t} B^{1/2} ds &\lesssim \int_{-T}^{t} \dfrac{\delta}{\dfrac{x(s)^2}{4}+\delta^2}dt \\
 &\lesssim \int_{-T}^{t} \dfrac{\delta}{\dfrac{\xi_0^2s^2}{16}+\delta^2}dt\\
 &\lesssim \int_{-\infty}^{t} \dfrac{1}{\delta}\dfrac{1}{\dfrac{\xi_0^2s^2}{16\delta^2}+1}dt\\
 &\lesssim \dfrac{4}{\xi_0}\int_{-\infty}^{\widetilde{t}} \dfrac{1}{r^2+1}dr,
\end{align*}
with $\widetilde{t}=\xi_0t/(4\delta)$. We finally find
\begin{align*}
  \int_{-T}^{t} B^{1/2} & \lesssim  \dfrac{4}{\xi_0} \left[\arctan\left(\dfrac{\xi_0t}{4\delta}\right)
  +\dfrac{\pi}{2} \right] \\
  & \lesssim -\arctan\left(\dfrac{4\delta}{\xi_0t}\right) \lesssim \dfrac{\sqrt{\eps}}{|t|}+ o(1).
\end{align*}
So, choosing $M$ big enough, we can write
\begin{equation}
 \label{lin2}
 \|\eps \d_x \varphi^\eps. d\chi_\delta\|_{L^1([-T,t];L^2)} \lesssim \dfrac{\sqrt{\eps}}{|t|}.
\end{equation}
\\[2mm]
In order to estimate the derivative of this term, we write
\begin{multline*}
 \left|\eps \d_x \left(\eps \d_x \varphi^\eps. d\chi_\delta\right)\right| \lesssim 
 \eps^{-1/4}\left(|\eps \d_x^2 u.d\chi_\delta| + |\sqrt{\eps}\d_x u.d\chi_\delta |+|u.d\chi_\delta|\right.\\
 \left.+ |\eps\sqrt{\eps}\d_x u.d^2\chi_\delta| + |\eps u.\delta^2 \chi_\delta|\right).
\end{multline*}
Using \eqref{eigenvectors}, and \eqref{lin2}, we argue as in the previous computations, keeping only ``worst'' terms and find
\begin{equation}
 \label{dlin2}
 \|\eps\d_x\left(\eps \d_x \varphi^\eps. d\chi_\delta\right)\|_{L^1([-T,t];L^2)} \lesssim \dfrac{\sqrt{\eps}}{|t|}.
\end{equation}
\noindent
\textbf{$\bullet$ Analysis of $\eps^{-1}\mathcal{R}^\eps_\delta\varphi^\eps \chi_\delta$}
\\[1mm]
\begin{align*}
\left|\mathcal{R}^\eps_\delta(t,x) \right|& =  \left|\dfrac{\lambda^{(3)}(g(t,x))}{6}(x-x(t))^3\right|\\
& \lesssim \dfrac{\delta^2(x-x(t))^4+\delta^2x(t)(x-x(t))^3}{\left[\left(\tau(x-x(t))+x(t)\right)^2+\delta^2\right]^{5/2}}.
\end{align*}
We notice that $\quad \displaystyle \dfrac{1}{\left(\left[\tau(x-x(t))+x(t) \right]^2+\delta^2 \right)^{5/2}}\lesssim
\dfrac{1}{\left(x(t)^2+\delta^2\right)^{5/2}}, \quad $ 
since 
\\ $x(t)=t.\xi_0+\mathcal{O}(t^2)$, for $T$ sufficiently small, 
$\exists C_1, C_2$ independent of $\delta$ (they depend on $T$) such that $$C_1|t|\leq |x(t)|\leq C_2|t|.$$
Then
$$ |\mathcal{R}^\eps_\delta(t,x) | \lesssim \dfrac{\delta^2(x-x(t))^4}{\left(t^2+\delta^2\right)^{5/2}} 
+ \dfrac{\delta^2t(x-x(t))^3}{\left(t^2+\delta^2\right)^{5/2}}.$$
We can write
\\
$\|\eps^{-1} \mathcal{R}^\eps_\delta \varphi^\eps \chi_\delta\|_{L^1([-T,t],L^2)} $
\begin{align*}
& \lesssim  \eps^{-1}\int_{-T}^{t} \left[\int_{x\in\R} \eps^{-1/2}|u(s,...)|^2
\left( \dfrac{\delta^4(x-x(s))^8}{\left(s^2+\delta^2\right)^{5}} 
+ \dfrac{\delta^4s^2(x-x(s))^6}{\left(s^2+\delta^2\right)^{5}} \right) dx\right]^{1/2} ds \\
& \lesssim \eps^{-1}\int_{-T}^{t} \dfrac{\delta^2\eps^2}{(s^2+\delta^2)^{5/2}} ds \times \|y^4 u \|_{L^\infty(L^2)} 
+ \eps^{-1}\int_{-T}^{t} \dfrac{s\delta^2\eps^{3/2}}{(s^2+\delta^2)^{5/2}} ds \times \|y^3 u \|_{L^\infty(L^2)}\\
&\lesssim C(u) \left(I + II\right).
\end{align*}
\begin{align*}
 I \lesssim \int_{-T}^{t}  \dfrac{\delta^2\eps}{(s^2+\delta^2)^{5/2}} ds & \lesssim
 \eps\delta^{-3} \int_{-T}^{t}  \dfrac{1}{\left(\left(\dfrac{s}{\delta}\right)^2+1\right)^{5/2}} ds\\
 & \lesssim \eps \delta^{-2}\int_{-T/\delta}^{t/\delta} \dfrac{1}{\left(r^2+1\right)^{5/2}} dr\\
 & \lesssim \eps \delta^{-2}\int_{-\infty}^{t/\delta} \dfrac{1}{r^5}dr\\
 & \lesssim \dfrac{\eps^2}{t^4},
\end{align*}
and
\begin{align*}
 II \lesssim \int_{-T}^{t}  \dfrac{s\delta^2\eps^{1/2}}{(s^2+\delta^2)^{5/2}} ds & \lesssim
 \eps^{1/2}\delta^{-3} \int_{-T}^{t}  \dfrac{s}{\left(\left(\dfrac{s}{\delta}\right)^2+1\right)^{5/2}} ds\\
 & \lesssim \eps^{1/2} \delta^{-1}\int_{-T/\delta}^{t/\delta} \dfrac{r}{\left(r^2+1\right)^{5/2}}dr\\
 & \lesssim \eps^{1/2} \delta^{-1}\int_{-\infty}^{t/\delta} \dfrac{r}{\left(r^2+1\right)^{5/2}}dr \\
 & \lesssim \dfrac{\eps^{3/2}}{|t|^3},
\end{align*}
which finally gives
\begin{equation}
 \label{lin3}
 \|\eps^{-1} \mathcal{R}^\eps_\delta \varphi^\eps \chi_\delta\|_{L^1([-T,t],L^2)} \lesssim \dfrac{\eps^{3/2}}{|t|^3}.
\end{equation}
\\[2mm]
For its derivative, we write
\begin{align*}\eps \d_x \left(\eps^{-1} \mathcal{R}^\eps_\delta \varphi^\eps \chi_\delta \right) & =
(\d_x \mathcal{R}^\eps_\delta) \varphi^\eps \chi_\delta 
+ \eps^{-1} \mathcal{R}^\eps_\delta \eps \d_x(\varphi^\eps) \chi_\delta
+ \mathcal{R}^\eps_\delta \varphi^\eps d\chi_\delta \\
& = (1) + (2) + (3).
\end{align*}

\noindent
 \textbf{Study of $(1)$:} $\left|\d_x \mathcal{R}^\eps_\delta \right|
 =\left|\dfrac{\lambda^{(3)}(g(t,x))}{2}(x-x(t))^2\right|+\dfrac{\tau}{6}|\lambda^{(4)}(g(t,x))(x-x(t))^2|$. 
 We write $g(t,x) = \tau x + (1-\tau)x(t), \; \tau \in ]0,1[$, and since $x(t)=t.\xi_0+\mathcal{O}(t^2)$, 
 for $T$ sufficiently small, $\exists C_1, C_2$ independent of $\delta$ (they depend on $T$) such that
$$C_1|t|\leq |x(t)|\leq C_2|t|,$$ we can write
 \begin{align*}
\left|\d_x \mathcal{R}^\eps_\delta \right| & \lesssim \dfrac{\delta^2(x-x(t))^5+\delta^2x(t)^2(x-x(t))^3+\delta^4(x-x(t))^3}
{\left[\left(\tau(x-x(t))+x(t)\right)^2+\delta^2\right]^{7/2}}\\
& \lesssim \dfrac{\delta^2(x-x(t))^5+\delta^2 t^2(x-x(t))^3+\delta^4(x-x(t))^3}
{\left[t^2+\delta^2\right]^{7/2}}.
 \end{align*}
Arguing as before, we find
\\
$ \|(\d_x \mathcal{R}^\eps_\delta) \varphi^\eps \chi_\delta\|_{L^1([-T,t],L^2)}$
\begin{align*}
 & \lesssim 
 C \left(\int_{-T}^{t} \dfrac{\delta^2\eps^{5/2}}{\left(s^2+\delta^2\right)^{7/2}} ds
 +\int_{-T}^{t} \dfrac{\delta^2s^2\eps^{3/2}}{\left(s^2+\delta^2\right)^{7/2}}ds 
 +\int_{-T}^{t} \dfrac{\delta^4\eps^{3/2}}{\left(s^2+\delta^2\right)^{7/2}}ds\right) \\
 & \lesssim C \left(I+II+III \right),
\end{align*}
where $C$ is a constant depending on $L^2-$norms of momenta of the profile. We write:
\begin{align*}
 I & \lesssim \int_{-T}^{t} \dfrac{\delta^2\eps^{5/2}}{\left(s^2+\delta^2\right)^{7/2}} ds \\
 & \lesssim \int_{-T/\delta}^{t/\delta} \dfrac{\delta^{-4}\eps^{5/2}}{\left(r^2+1\right)^{7/2}} dr \\
 & \lesssim \delta^{-4}\eps^{5/2}\int_{-\infty}^{t/\delta} \dfrac{1}{r^7} dr
 \lesssim 
 \dfrac{\eps^{7/2}}{t^6};
\end{align*}
and
\begin{align*}
II & \lesssim \int_{-T}^{t} \dfrac{\delta^2s^2\eps^{3/2}}{\left(s^2+\delta^2\right)^{7/2}}ds \\
 & \lesssim \int_{-T/\delta}^{t/\delta} \dfrac{\delta^{-2}r^2\eps^{3/2}}{\left(r^2+1\right)^{7/2}}dr\\
 & \lesssim \delta^{-2}\eps^{3/2}\int_{-\infty}^{t/\delta} \dfrac{r^2}{r^2+1}\times \dfrac{1}{\left(r^2+1\right)^{5/2}}dr\\
 & \lesssim \delta^{-2}\eps^{3/2}\int_{-\infty}^{t/\delta} \dfrac{1}{r^5}dr \lesssim \dfrac{\eps^{5/2}}{t^4}.
\end{align*}
Using same computations as for $I$, we find for $III$
\begin{equation*}
 III \lesssim \int_{-T}^{t} \dfrac{\delta^4\eps^{3/2}}{\left(s^2+\delta^2\right)^{7/2}}ds \lesssim \dfrac{\eps^{7/2}}{t^6}
\end{equation*}
So, we only keep the worst term and find: 
\begin{equation*}
 \|(\d_x \mathcal{R}^\eps_\delta) \varphi^\eps \chi_\delta\|_{L^1([-T,t],L^2)} \lesssim \dfrac{\eps^{5/2}}{t^4}.
\end{equation*}
\\[1mm]
\noindent
\textbf{Study of $(2)$:} We then write 
\begin{align*}
  \left|\eps^{-1}\mathcal{R}^\eps_\delta \eps \d_x(\varphi^\eps) \right| 
  & = \left|\eps^{-1}\mathcal{R}^\eps_\delta \eps^{-1/4}\left(\sqrt{\eps}\d_x u(t,\ldots)+ i \xi(t)u(t,\ldots) \right)\right|\\
  &\lesssim \left|\eps^{-1}\mathcal{R}^\eps_\delta \eps^{-1/4}\sqrt{\eps}\d_x u(t,\ldots)\right| 
  + \left|\eps^{-1}\mathcal{R}^\eps_\delta \varphi^\eps\right|.
\end{align*}
 Using \eqref{lin3}, and since $\|y^\alpha \d_x^\beta u(t)\|_{L^2}$ are bounded for all $\alpha, \beta$, we can easily find
 \begin{align*}
  \|\eps^{-1} \mathcal{R}^\eps_\delta \eps\d_x\varphi^\eps \chi_\delta\|_{L^1([-T,t],L^2)} & \lesssim 
  \sqrt{\eps} \|\eps^{-1} \mathcal{R}^\eps_\delta \eps^{-1/4}\d_x u\|_{L^1([-T,t],L^2)}\\
  & \quad+\|\eps^{-1} \mathcal{R}^\eps_\delta \eps\d_x\varphi^\eps\|_{L^1([-T,t],L^2)}\\
  & \lesssim \dfrac{\eps^{3/2}}{|t|^3}.
 \end{align*}
 \\[1mm]
\noindent
\textbf{Study of $(3)$:} We finally use \eqref{eigenvectors} and see 
 $$\left|\mathcal{R}^\eps_\delta \varphi^\eps d\chi_\delta \right|\lesssim \eps\delta^{-1} 
 \left| \eps^{-1}\mathcal{R}^\eps_\delta \varphi^\eps\right|.$$
 Using \eqref{lin3}, we can write:
 \begin{align*}
  \|\mathcal{R}^\eps_\delta \varphi^\eps d\chi_\delta\|_{L^1([-T,t],L^2)} & \lesssim 
  \eps\delta^{-1} \|\eps^{-1}\mathcal{R}^\eps_\delta \varphi^\eps\|_{L^1([-T,t],L^2)} \\
  & \lesssim \dfrac{\eps^2}{|t|^3}.
 \end{align*}
 So we obtain
 \begin{equation}
 \label{dlin3}
 \|\eps \d_x \left(\eps^{-1} \mathcal{R}^\eps_\delta \varphi^\eps \chi_\delta\right)\|_{L^1([-T,t],L^2)} 
 \lesssim \dfrac{\eps^{3/2}}{|t|^3}.
\end{equation}
 \\[1mm]
\noindent \textbf{Conclusion:} Combining \eqref{lin1},\eqref{lin2} and \eqref{lin3}, and dropping better contributions, we obtain
\begin{equation}
 \label{L}
 \|L^\eps\|_{L^1([-T,t],L^2)}\lesssim \dfrac{\sqrt{\eps}}{|t|},
\end{equation}
and thanks to \eqref{dlin1}, \eqref{dlin2} and \eqref{dlin3}:
\begin{equation}
 \label{dL}
 \|\eps \d_xL^\eps\|_{L^1([-T,t],L^2)}\lesssim \dfrac{\sqrt{\eps}}{|t|},
\end{equation}
and taking $C_L$ as the largest constant in front of these terms, we obtain Lemma~\ref{Leps}.

\subsection{Nonlinear terms and end of the proof}
\label{endproofw}
We recall the equation \eqref{w} satisfied by the remainder $w^\eps$:
\begin{equation*}
\left\lbrace \begin{array}{l} 
               i\eps\d_t w^\eps+ \dfrac{\eps^2}{2}\d_x^2 w^\eps - V_\delta(x) w^\eps = \eps {NL}^\eps+ \eps {L}^\eps; \\
               w^\eps(-T,x)= 0
\end{array}\right. ,
\end{equation*}
where
\begin{align*}
{NL}^\eps & = \kappa \eps^{1/2} \left(|\psi^\eps|^2\psi^\eps - |\varphi^\eps|^2 \varphi^\eps \chi_\delta \right) \\
{L}^\eps & = \eps^{-1}\mathcal{R}^\eps_\delta(t,x)\varphi^\eps \chi_\delta + \eps \d_x \varphi^\eps. d\chi_\delta 
 + \eps/2 \varphi^\eps d^2 \chi_\delta.
\end{align*}
Thanks to the Duhamel formula, we use a standard $L^2-$estimate and find
\begin{equation*}
\|w^\eps(t)\|_{L^2} \leq \|w^\eps(-T)\|_{L^2} + \int_{-T}^{t} \|NL^\eps(s)\|_{L^2}ds
+ \int_{-T}^{t}\|L^\eps(s)\|_{L^2}ds.
\end{equation*}
We now focus on the nonlinear terms; we have the following pointwise estimate:
\begin{equation*}
 \left||\psi^\eps|^2\psi^\eps - |\varphi^\eps|^2 \varphi^\eps \chi_\delta \right| 
 \lesssim \left(|w^\eps|^2+|\varphi^\eps|^2 \right)|w^\eps|,
\end{equation*}
and so
\begin{align*}
 \|NL^\eps\|_{L^2}& \lesssim \sqrt{\eps} \left\|\left(|w^\eps|^2+|\varphi^\eps|^2\right)|w^\eps| \right\|_{L^2} \\
 & \lesssim \sqrt{\eps}\left( \|w^\eps\|_{L^\infty}^2+\|\varphi^\eps\|_{L^\infty}^2\right) \|w^\eps\|_{L^2}.
\end{align*}
We recall that thanks to Corollary \ref{phiLinfty}, we have $$\|\varphi^\eps(t) \|_{L^\infty} \lesssim \eps^{-1/4},$$
and we perform the following bootstrap argument:
\begin{equation}
 \label{bootstrap}
 \|w^\eps(t)\|_{L^\infty}\leq \dfrac{M}{|t|}, \quad M>0.
\end{equation}
So
\begin{equation*}
 \|NL^\eps(s)\|_{L^2} \leq K\left(\dfrac{M^2\sqrt{\eps}}{s^2}+1 \right)\|w^\eps(s)\|_{L^2},
\end{equation*}
and we infer by Lemma \ref{Leps}
\begin{equation*}
 \|w^\eps(t)\|_{L^2} \leq C_L \dfrac{\sqrt{\eps}}{|t|} + 
 K \int_{-T}^{t} \left(\dfrac{M^2\sqrt{\eps}}{s^2}+1 \right)\|w^\eps(s)\|_{L^2} ds.
\end{equation*}
Using Gronwall lemma we obtain
\begin{equation*}
 \|w^\eps(t)\|_{L^2} \leq C_L \dfrac{\sqrt{\eps}}{|t|} - KC_L \int_{-T}^{t} \left(\dfrac{\eps M^2}{s^3}+\dfrac{\sqrt{\eps}}{s} \right)
 \exp\left({K\int_{-T}^{s}\left( 1+ \dfrac{\sqrt{\eps}M^2}{r^2}\right)dr}\right) ds
\end{equation*}
We notice that for $-T\leq t \leq -\Lambda \sqrt{\eps}$
\begin{align*}
 \exp\left({K\int_{-T}^{s}\left( 1+ \dfrac{\sqrt{\eps}M^2}{r^2}\right)dr}\right)& =
 e^{K\left(T-\frac{\sqrt{\eps}M^2}{T}\right)}e^{K\left(s-\frac{\sqrt{\eps}M^2}{s}\right)}\\
 & \leq C_T e^{K\left(\frac{M^2}{\Lambda}-\Lambda\sqrt{\eps}\right)}\leq C_T e^{\frac{KM^2}{\Lambda}},
\end{align*}
and this bound is independent of $\eps$, $t$ and $s$. We then write
\begin{align*}
 \|w^\eps(t)\|_{L^2} & \leq C_L \dfrac{\sqrt{\eps}}{|t|} - 
 C_T e^{\frac{KM^2}{\Lambda}}\; KC_L \int_{-T}^{t} \left(\dfrac{\eps M^2}{s^3}+\dfrac{\sqrt{\eps}}{s} \right)\\
 & \leq C_L \dfrac{\sqrt{\eps}}{|t|} + C_T e^{\frac{KM^2}{\Lambda}}\;
 KC_L \left[\sqrt{\eps} \log(|t|)+ \dfrac{\eps M^2}{2t^2} \right],
\end{align*}
where we have dropped all terms in $T$ since they present a $\mathcal{O}(\sqrt{\eps})$ contribution. 
Since $|t| \geq \Lambda \sqrt{\eps}$ we finally find
\begin{equation}
 \label{wL2}
 \|w^\eps(t)\|_{L^2} \leq C_L \dfrac{\sqrt{\eps}}{|t|} + KC_L\;C_T e^{\frac{KM^2}{\Lambda}} \dfrac{M^2}{2\Lambda}\;
 \dfrac{\sqrt{\eps}}{|t|}.
\end{equation}
\\[2mm]
It remains to check the validity of the bootstrap assumption \eqref{bootstrap}. We study the $\eps-$derivative of $w^\eps$:
\begin{equation*}
\left\lbrace \begin{array}{l} 
               i\eps\d_t (\eps \d_x w^\eps)+ \dfrac{\eps^2}{2}\d_x^2 (\eps \d_x w^\eps) - V_\delta(x)(\eps \d_x w^\eps) = 
               (\eps \d_x V_\delta) w^\eps + \eps^2 \d_x (NL^\eps+ L^\eps); \\
               \eps \d_x w^\eps(-T,x)= 0,
\end{array}\right.
\end{equation*}
A standard $L^2-$ estimate allows us to find
\begin{align*}
 \|\eps \d_x w^\eps(t)\|_{L^2} & \lesssim \|\eps \d_x w^\eps(-T)\|_{L^2} + \int_{-T}^{t} \|\eps \d_x NL^\eps (s)\|_{L^2}ds
 + \int_{-T}^{t} \|\eps \d_x L^\eps (s)\|_{L^2}ds\\
 & \quad + \int_{-T}^{t} \|(\eps \d_x V_\delta) w^\eps(s)\|_{L^2}ds.
\end{align*}
We first notice that since $|\d_x V_\delta(x)| \leq C$ for all $x\in \R$, where $C$ is independent of $\delta$ and $x$; 
we have, using \eqref{wL2} and keeping only the worst contributions:
\begin{equation*}
 \|(\eps \d_x V_\delta) w^\eps\|_{L^1([-T,t],L^2)} \lesssim 
 \eps \sqrt{\eps} \left[C_L+K C_L C_T e^{\frac{KM^2}{\Lambda}}\dfrac{M^2}{2\Lambda}\right] \log |t|,
\end{equation*}
Then, we write
\begin{align*}
 \left|\eps \d_x NL^\eps\right| & \lesssim \sqrt{\eps}
 \left| |w^\eps + \varphi^\eps \chi_\delta|^2 ( \eps \d_x w^\eps + \eps \d_x \varphi^\eps \chi_\delta 
 + \varphi^\eps \eps d\chi_\delta)\right. \\
 & \quad \left.- |\varphi^\eps|^2 \left( \eps \d_x \varphi^\eps \chi_\delta + \varphi^\eps \eps d\chi_\delta\right)\right| \\
 & \lesssim  \sqrt{\eps} \left|w^\eps+\varphi^\eps \chi_\delta \right|^2 |\eps \d_x w^\eps| \\
 & \quad + \sqrt{\eps}\left(|w^\eps+\varphi^\eps\chi_\delta|^2-|\varphi^\eps|^2 \right) 
 \left|\eps \d_x \varphi^\eps \chi_\delta + \varphi^\eps \eps d\chi_\delta \right|.
\end{align*}
The first part will be handled as before. For the second part of this term, we write
\\
 $\left||w^\eps+\varphi^\eps\chi_\delta|^2-|\varphi^\eps|^2 \right| \times 
 \left|\eps \d_x \varphi^\eps \chi_\delta + \varphi^\eps \eps d\chi_\delta \right| $
 \begin{align*}
 & \lesssim 
 \left(|w^\eps|+ |\varphi^\eps| \right)|w^\eps|\; |\eps \d_x \varphi^\eps | 
 + \left(|w^\eps|+ |\varphi^\eps| \right)|w^\eps|\; |\varphi^\eps\eps d\chi_\delta | \\
 &\lesssim \left(|w^\eps|+ |\varphi^\eps| \right)|w^\eps|\; |\eps \d_x \varphi^\eps | 
 + \left(|w^\eps|+ |\varphi^\eps| \right)|w^\eps|\; \left|\varphi^\eps\dfrac{\eps}{\delta} \right|,
\end{align*}
where we have used \eqref{eigenvectors}.
So 
\begin{align*}
 \|\eps \d_x NL^\eps (s)\|_{L^2} & \lesssim 
 \sqrt{\eps}\left( \|w^\eps\|_{L^\infty}^2+\|\varphi^\eps\|_{L^\infty}^2\right)\|\eps \d_x w^\eps\|_{L^2}  \\
 & \quad + \sqrt{\eps} \left( \|w^\eps\|_{L^\infty}+\|\varphi^\eps\|_{L^\infty}\right)
 \|\eps\d_x \varphi^\eps\|_{L^\infty}\|w^\eps\|_{L^2}\\
 & \quad + \sqrt{\eps} \left( \|w^\eps\|_{L^\infty}+\|\varphi^\eps\|_{L^\infty}\right)
 \left\|\dfrac{\eps}{\delta}\varphi^\eps\right\|_{L^\infty}\|w^\eps\|_{L^2}\\
 & \lesssim \left[\dfrac{M^2\sqrt{\eps}}{s^2}+1 \right] \|\eps \d_x w^\eps(s) \|_{L^2} \\
 & \quad + \left(\dfrac{\eps^{3/4}M}{s^2}+\dfrac{\sqrt{\eps}}{|s|} \right)\; C_L
 \left( 1 + K\;C_T e^{\frac{KM^2}{\Lambda}} \dfrac{M^2}{2\Lambda}\right),
\end{align*}
where we have used Corollary \ref{phiLinfty}, \eqref{bootstrap}, and \eqref{wL2}. 
We recall that by Lemma \ref{Leps}, we have
$$\| \eps \d_x L^\eps\|_{L^1([-T,t],L^2)} \leq C_L \dfrac{\sqrt{\eps}}{|t|}, $$
and it is easily seen that this contribution is the most important. So we find, dropping better terms
\begin{equation*}
 \|\eps \d_x w^\eps(t)\|_{L^2} \leq C_L \dfrac{\sqrt{\eps}}{|t|} + 
 \widetilde{K}\int_{-T}^{t} \left(1+ \dfrac{\sqrt{\eps}M^2}{s^2} \right)ds.
\end{equation*}
Applying Gronwall lemma again, and arguing as before, as long as \eqref{bootstrap} holds, we finally get
\begin{equation}
 \label{wH1}
 \|\eps \d_x w^\eps(t)\|_{L^2} \leq C_L \dfrac{\sqrt{\eps}}{|t|} 
 + \widetilde{K}\widetilde{C_T}C_L e^{\frac{\widetilde{K}M^2}{\Lambda}}\; \dfrac{M^2}{2\Lambda}\;\dfrac{\sqrt{\eps}}{|t|}.
\end{equation}
Note that $t$ has to satisfy $\eps^{1/2-\gamma} \lesssim |t|$, for $0<\gamma<1/2$, 
if we want $w^\eps$ to become smaller in $H^1_\eps(\R)$, as $\eps$ tends to zero.
\\
\noindent
We can now use Gagliardo-Nirenberg inequality, and obtain thanks to \eqref{wL2} and \eqref{wH1}
\begin{align*}
 \|w^\eps(t)\|_{L^\infty} & \leq C_{GN} \; \eps^{-1/2}\|w^\eps(t)\|_{L^2}^{1/2}\; \|\eps \d_x w^\eps (t)\|_{L^2}^{1/2}\\
 & \leq \dfrac{C_L}{|t|} \left[C_1+C_1\; e^{C_2\frac{M^2}{\Lambda}}\dfrac{M^2}{\Lambda} \right],
\end{align*}
where $C_1$ and $C_2$ are both nonnegative constants, depending on $K, \widetilde{K}, C_T, \widetilde{C_T}, C_{GN}$ only. 
Taking $$M=2C_1\; C_L, $$ and choosing $\Lambda$ large enough to have
$$e^{C_2\frac{M^2}{\Lambda}}\dfrac{M^2}{\Lambda} \leq \dfrac{1}{2}, $$
we have
\begin{equation*}
 \|w^\eps(t)\|_{L^\infty}  \leq \dfrac{C_L}{|t|}\; \dfrac{3}{2}C_1 \leq \dfrac{3}{4}\;\dfrac{M}{|t|},
\end{equation*}
so the remainder $w^\eps$ is small in the appropriate spaces and the bootstrap assumption \eqref{bootstrap} holds 
for $-T\leq t \leq -c_0 \eps^{1/2-\gamma}$, for any $c_0>0$ and $0<\gamma<1/2$. 
And the proof is complete. 
\\[5mm]
\noindent
\textbf{Conclusion: }Since this approximation is valid until a time $-t^\eps = -c_0 \eps^{1/2-\gamma}$, where $c_0$ is a fixed constant, 
for any $\gamma \in ]0, 1/2[$, we can choose the exponent $\gamma$ as small as we wish. 
In order to study the propagation through the crossing point, 
we can now work on a time interval of the form $[-c_0 \eps^{1/2-\gamma}, c_0 \eps^{1/2-\gamma}]$ without restriction,
as it is done in the linear case, in \cite{HJ98}.

\section{Approximation in the crossing region.}
\label{inner}
This section is devoted to the proof of Theorem \ref{thm:Inner}.
\subsection{Introduction of a new regime}
\label{match}
We recall the free trajectories, introduced in Section \ref{introInnerRegime} 
 $$ \widetilde{\xi}(t)=\xi_0>0 \quad ; \quad  \widetilde{x}(t) = \xi_0t,$$
and the following rescaled variables:
$$\left\lbrace \begin{array}{c}
y = \left(x-t\xi_0\right)/\sqrt{\eps}\\
                s=t/\sqrt{\eps}.
               \end{array}\right. $$
We then recall that for $t\in [-c_0 \eps^{1/2-\gamma}, c_0 \eps^{1/2-\gamma}]$, where $0< \gamma < 1/6$, $v^\eps$ is defined by
$$\psi^\eps(t,x)= \eps^{-1/4}v^\eps \left(\dfrac{t}{\sqrt{\eps}}, \dfrac{x-t\xi_0}{\sqrt{\eps}} \right)
e^{\frac{i\xi_0^2t}{2\eps}+\frac{i\xi_0.(x-t\xi_0)}{\eps}}, $$
and the associated Schr\"odinger equation \eqref{NLS1} is
\begin{equation*}
\left\lbrace \begin{array}{l}
               i\d_s v^\eps - V_{\delta/\sqrt{\eps}}\left(y+s\xi_0\right) v^\eps = 
 -\dfrac{\sqrt{\eps}}{2}\d_y^2 v^\eps + \kappa \sqrt{\eps} |v^\eps|^2v^\eps, \\
              v^\eps\left(-s^\eps,y\right)=v^\eps_{init.}(y),
             \end{array}
 \right.
\end{equation*}
where $s^\eps=c_0\eps^{-\gamma}$ and
$$V_{\delta/\sqrt{\eps}}\left(y+s\xi_0\right)= \begin{pmatrix}
                                                y+s\xi_0 & c \\
                                                c & -(y+s\xi_0)
                                               \end{pmatrix}.$$ 
\\[2mm]
Our aim is to construct an approximation of $v^\eps$ on $[-s^\eps,s^\eps]$. 
In the linear case, the solution is approached by $f$, solution to \eqref{f}-\eqref{dataf}:
\begin{equation*}
\left\lbrace \begin{array}{l}
  i\d_s f - \left(  \begin{array}{cc}
            y+s\xi_0 & c \\
            c & -(y+s\xi_0)
           \end{array}
\right)f = 0 \\
f(-s^\eps,y)=u_\delta(-\sqrt{\eps}s^\eps,y)\left(\begin{array}{c}
                                                                  0 \\1
                                                                 \end{array}
 \right)e^{\frac{i}{\eps}\phi^\eps(y)}.
\end{array}\right.
\end{equation*}
In fact, assuming the approximation has an expansion of the form
$$v^\eps_{app}(s,y) = \sum_{j=0}^{\infty}c_j\eps^{j/2}f_j(s,y), $$
where $c_j$ and $f_j$ are vector-valued constants and functions, we can easily see that the lowest order terms satisfies \eqref{f}. 
Since our power of $\eps$ in front of the nonlinearity is 1/2, whereas there are smaller powers of $\eps$, 
we claim that the nonlinearity does not affect the evolution at leading order. So we introduce $r^\eps$ the difference
$$r^\eps(s,y) = v^\eps(s,y)-f(s,y),$$
which satisfies
\begin{equation*}
   \left\lbrace \begin{array}{l} i\d_s r^\eps + \dfrac{\sqrt{\eps}}{2}\d_y^2 r^\eps - \begin{pmatrix}
                                                        y+s\xi_0 & c \\
                                                        c & -(y+s\xi_0)
                                                       \end{pmatrix}r^\eps = \dfrac{-\sqrt{\eps}}{2}\d_y^2 f
                                                       +\sqrt{\eps}|r^\eps+f|^2(r^\eps+f)
                                                       \\
                                                       \\
                                                       r^\eps(-s^\eps,y)=v^\eps_{init.}(y)-f(-s^\eps,y).
  \end{array}\right.\end{equation*}
 \noindent                   
The aim is to prove that this term is small in $L^2$ in the semiclassical limit. We first need to study $v^\eps_{init.}$ 
in order to match it with the approximation $\varphi^\eps$ at the initial time with the new rescaled variables. 
Then we need some results on the behaviour of $f$, and to prove the validity of it as an approximation at leading order of $v^\eps$, 
which will allow us to deduce the behaviour of the exact solution on the time interval that we consider.                     
\\
This section is divided into three parts:
\begin{enumerate}
 \item Analysis of the initial data
 \item Analysis of $f$
 \item Proof of Theorem \ref{thm:Inner}
\end{enumerate}

\subsection{Initial data and matching}  
The aim of this section is to prove the following proposition:
\begin{proposition}
 \label{prop:match}
 Let $v^\eps$ satisfying \eqref{rescaled}. Then $v^\eps$ solves \eqref{NLS1}
 with  $$v^\eps_{init.}\left (y\right) = u_\delta\left(-c_0\eps^{1/2-\gamma},y\right)e^{\frac{i\phi^\eps(y)}{\eps}}
 \left( \begin{array}{c}
         0 \\ 1
        \end{array}\right) + W^\eps(y), $$
where $u_\delta$ is the profile, solution to \eqref{profile}, the phase is the difference between both phases 
\begin{align}
\label{phasephi}
\dfrac{i}{\eps}\phi^\eps(y)   &=\dfrac{i}{\eps}\left(S(-\sqrt{\eps}s^\eps)+\dfrac{\sqrt{\eps}}{2} \xi_0^2 s^\eps\right) \\
\notag
 & \quad +\dfrac{i}{\eps}\left[\xi(-\sqrt{\eps}s^\eps).
 (\sqrt{\eps}y-\sqrt{\eps}s^\eps\xi_0-x(-s^\eps\sqrt{\eps}))-\xi_0\sqrt{\eps}y\right],
 \end{align}
and where $W^\eps$ satisfies
$$\|W^\eps \|_{L^2}+\|\sqrt{\eps}\d_y W^\eps \|_{L^2} \lesssim \eps^{\gamma}.$$
\end{proposition}
\begin{proof}
We first notice that thanks to Theorem \ref{thm:IncomingOuter}, we have for $\gamma \in ]0,1/6[$,
$$\psi^\eps(-t^\eps,x)=\eps^{-1/4}u_\delta\left(-t^\eps, \dfrac{x-x(-t^\eps)}{\sqrt{\eps}}\right)
e^{\frac{i}{\eps}\Phi^\eps(-t^\eps,x)}\chi_\delta(x) + w^\eps(-t^\eps,x),$$                                               
where $t^\eps=c_0\eps^{1/2-\gamma}$, and
$\Phi^\eps(-t^\eps,x)=S^+(-t^\eps)+\xi^+(-t^\eps).(x-x^+(-t^\eps))$. But, since
$$\psi^\eps(-t^\eps,x)=\eps^{-1/4}v^\eps\left(\dfrac{-t^\eps}{\sqrt{\eps}},\dfrac{x+t^\eps\xi_0}{\sqrt{\eps}}\right)
e^{\frac{i}{\eps}\widetilde{\Phi}^\eps(-t^\eps,x)}, $$
with $\widetilde{\Phi}^\eps(t,x) = \xi_0^2 t/2+\xi_0(x-t\xi_0)$, 
we have to match these two expressions of the same term. 
\\[2mm]
In the rescaled variables $(s,y)$, we have for $-s^\eps=-c_0\eps^{-\gamma}$:
\begin{equation*}
v^\eps(-s^\eps,y)=u_\delta \left(-\sqrt{\eps}s^\eps,y+\theta(\eps) \right)\chi_\delta(\sqrt{\eps}(y-s^\eps\xi_0))
e^{\frac{i}{\eps}\phi^\eps(y)} + \omega^\eps\left(-s^\eps,y\right), 
\end{equation*}
where 
\begin{itemize}
\item $\theta(\eps)=-s^\eps\xi_0-\dfrac{x(-s^\eps\sqrt{\eps})}{\sqrt{\eps}}$,
\item the phase is given by \eqref{phasephi}:
\begin{align*}
\dfrac{i}{\eps}\phi^\eps(y) &= \dfrac{i}{\eps}(\Phi^\eps-\widetilde{\Phi}^\eps)
\left(-\sqrt{\eps}s^\eps,\sqrt{\eps}(y-s^\eps\xi_0)\right) 
\\
  &=\dfrac{i}{\eps}\left(S(-\sqrt{\eps}s^\eps)+\dfrac{\sqrt{\eps}}{2}\xi_0^2 s^\eps\right) \\
 & \quad +\dfrac{i}{\eps}\left[\xi(-\sqrt{\eps}s^\eps).
 (\sqrt{\eps}y-\sqrt{\eps}s^\eps\xi_0-x(-s^\eps\sqrt{\eps}))-\xi_0\sqrt{\eps}y\right]
 \end{align*}
 \item $\omega^\eps$ is obtained writing
 $$w^\eps(t,x)=\eps^{-1/4}\omega^\eps\left(\dfrac{t}{\sqrt{\eps}},\dfrac{x-t\xi_0}{\sqrt{\eps}} \right)
 e^{\frac{i}{\eps}\widetilde{\Phi}(t,x)}, $$ where $w^\eps$, 
 the difference between the exact solution and the approximation is introduced in Section \ref{subsec:profile}, by \eqref{w}.
\end{itemize}
Thus we can write
\begin{equation}
 \label{datav}
 v^\eps(-c_0\eps^{-\gamma},y) = u_\delta (-c_0\eps^{1/2-\gamma},y)e^{\frac{i}{\eps}\phi^\eps(y)}\left(\begin{array}{c}
                                                                                                   0\\1
                                                                                                  \end{array}\right)
+W^\eps(y),
\end{equation}
where
\begin{align*}
& W^\eps(y) = \omega^\eps(-c_0\eps^{-\gamma},y)\\
& \quad + \theta(\eps)\left(\int_{0}^{1}\d_y u_\delta (-c_0\eps^{1/2-\gamma},\zeta \theta(\eps)+y)\;d\zeta \right) 
e^{\frac{i}{\eps}\phi^\eps(y)} \; 
\chi_\delta^+(\sqrt{\eps}(y-c_0\xi_0\eps^{-\gamma}))\\
& \quad +u_\delta(-c_0\eps^{1/2-\gamma},y)e^{\frac{i}{\eps}\phi^\eps(y)}
\left[\chi_\delta^+(\sqrt{\eps}(y-c_0\eps^{-\gamma}))-\left(\begin{array}{c}
 0\\1
\end{array}\right)\right].
\end{align*}
Our next steps are
\begin{description}
 \item[Step $(1)$]To check that $\theta(\eps)$ is small. 
 \item[Step $(2)$]To prove that $\phi^\eps$, the difference between both phases is not troublesome.
 \item[Step $(3)$]To prove that the three last terms of \eqref{datav} are small in the semi-classical limit, 
 and in the appropriate space.
\end{description}
\vspace{2mm}
\textbf{Step $(1)$: Analysis of $\theta(\eps)$.} Using \eqref{trajectories}, we compute higher derivatives of the classical trajectories and the action, and find 
those we will need:
\begin{align*}
   & \dot{x}=\xi;\quad \ddot{x}=\dot{\xi}= \dfrac{-x}{\sqrt{x^2+\delta^2}}=\mathcal{O}(1) \\
   &x^{(3)}= \ddot{\xi}= \dfrac{-\dot{x}\delta^2}{(x^2+\delta^2)^{3/2}}=\mathcal{O}\left(\dfrac{1}{\delta}\right)\\
   &\dot{S}=\dfrac{|\xi|^2}{2}-\sqrt{x^2+\delta^2};  \quad \ddot{S}= 2\xi \dot{\xi}=\mathcal{O}(1) \\
   & S^{(3)}=2\dot{\xi}^2+2\xi\ddot{\xi} = \mathcal{O}\left(\dfrac{1}{\delta}\right),
  \end{align*}
and so
\begin{equation}
\label{taylor}
\begin{array}{l|c|l}
 x(0)=0 & \xi(0)=\xi_0 & S(0)=0 \\
 \dot{x}(0)=\xi_0 & \dot{\xi}(0)=0 & \d_t S(0) = \dfrac{\xi_0^2}{2}-\delta \\
 x^{(2)}(0)=0 & \xi_0^{(2)}= \dfrac{-\xi_0}{\delta} & \d_t^2 S(0)=0 \\
 x^{(3)}(0)=\dfrac{-\xi_0}{\delta}& \xi^{(3)}(0)=x^{(4)}(0)=0 & \d_t^3 S(0)= \dfrac{-2\xi_0^2}{\delta}.
\end{array}.\end{equation}
An easy computation, using Taylor expansions of the classical trajectories around $-\sqrt{\eps}s^\eps$, which is small, 
allows us to write for $\tau=\tau(\eps)\in ]0,1[$:
\begin{equation}
\label{theta}
\theta(\eps) = -s^\eps\xi_0 -\dfrac{1}{\sqrt{\eps}} \left[-\xi_0 s^\eps\sqrt{\eps}
+ \dot{\xi}(-\tau\sqrt{\eps}s^\eps).\dfrac{\eps}{2}(s^\eps)^2\right] = \mathcal{O}(\eps^{\frac{1}{2}-2\gamma}),
\end{equation}
where we choose $\gamma$ small enough to have $1/2-2\gamma>0$. 
\\[2mm]
\textbf{Step $(2)$: Analysis of the phase $\phi^\eps(y)$.} We prove the following Lemma:
\begin{lemma}
 We consider $\phi^\eps$ the difference between both phases, given by \eqref{phasephi}. Then 
 \begin{equation}
  \label{phase}
  \dfrac{i}{\eps}\phi^\eps(y) = -ic_0\eps^{-\gamma}\dot{\xi}(-c_0\eps^{1/2-\gamma}\tau').y  + \mathcal{O}(\eps^{-3\gamma})
  = \mathcal{O}(\eps^{-\gamma}).y+\mathcal{O}(\eps^{-3\gamma}),
 \end{equation}
where the last term does not depend on $y$, and where $\dot{\xi}(-c_0\eps^{1/2-\gamma}\tau')=\mathcal{O}(1)$ for a 
$\tau'=\tau'(\eps)\in ]0,1[$.
\\
We deduce that 
\begin{equation}
\label{phase2}
\sqrt{\eps} \dfrac{d}{dy} \left(\dfrac{i}{\eps}\phi^\eps(y)\right) = \mathcal{O}(\eps^{\frac{1}{2}-\gamma}).
\end{equation}
\end{lemma}
\begin{proof}
We need to estimate the difference of both phases $\phi^\eps$ at time $s=-s^\eps$. This term will depend on 
$\eps$ and $y$ and will induce a loss of $\eps$ at each derivative for the initial data. In fact, we compute:
\begin{align*}\dfrac{i}{\eps}\phi^\eps(y)
  &=\dfrac{i}{\eps}\left(S(-\sqrt{\eps}s^\eps)-\xi_0^2 (-s^\eps)\sqrt{\eps}/2\right) \\
 & \quad +\dfrac{i}{\eps}\left[\xi(-\sqrt{\eps}s^\eps).
 (\sqrt{\eps}y-\sqrt{\eps}s^\eps\xi_0-x(-s^\eps\sqrt{\eps}))-\xi_0\sqrt{\eps}y\right]\\
 & = (i)+(ii).
 \end{align*}
We write, using Taylor expansion of the classical action, and \eqref{taylor}
\begin{align*}
 (i) & = \dfrac{i}{\eps}\left(S(-\sqrt{\eps}s^\eps)-\xi_0^2 (-s^\eps)\sqrt{\eps}/2\right) \\
 & = \dfrac{i\delta s^\eps\sqrt{\eps}}{\eps}+\dfrac{ig^\eps(-s^\eps)}{\eps}\\
 & = ic \; c_0 \eps^{-\gamma}+ \dfrac{ig^\eps(-s^\eps)}{\eps},
\end{align*}
where $$g^\eps(s)=\dfrac{S^{(3)}(\zeta)}{6}.(s\sqrt{\eps})^3, \quad \zeta = \tau s\sqrt{\eps}; \; \tau=\tau(\eps)\in]0,1[, \quad \textrm{and}\;
\left|S^{(3)}(\zeta)\right|\lesssim \dfrac{1}{\delta}.$$
We deduce
\begin{equation}
 \label{i}
 (i) = \mathcal{O}(\eps^{-3\gamma}).
\end{equation}
We then study the second term:
\begin{align}
\notag
 (ii) &= \dfrac{i}{\eps}\left[\xi(-\sqrt{\eps}s^\eps).
 (\sqrt{\eps}y-\sqrt{\eps}s^\eps\xi_0-x(-s^\eps\sqrt{\eps}))-\xi_0\sqrt{\eps}y\right]\\
 \notag
 &= \dfrac{i}{\sqrt{\eps}}\left(\xi(-s^\eps\sqrt{\eps})-\xi_0 \right).y+ \dfrac{i}{\eps}\xi(-s^\eps\sqrt{\eps}).
 \left(-s^\eps \sqrt{\eps}\xi_0-x(-s^\eps \sqrt{\eps}) \right)\\
 \notag
 & = \dfrac{i}{\sqrt{\eps}} (-s^\eps\sqrt{\eps}).\dot{\xi}(-s^\eps\sqrt{\eps}\tau').y 
 + \dfrac{i}{\eps}
 \left[\xi_0+\dot{\xi}(-s^\eps\sqrt{\eps}\tau')\right].\dot{\xi}(-s^\eps\sqrt{\eps}\tau'').\dfrac{(s^\eps)^2.\eps}{2}\\
 \label{ii}
 & = -ic_0 \eps^{-\gamma}\dot{\xi}(-s^\eps\sqrt{\eps}\tau').y +\mathcal{O}(\eps^{-2\gamma}) 
 = \mathcal{O}(\eps^{-\gamma}).y+\mathcal{O}(\eps^{-2\gamma}),
\end{align}
where $0<\tau',\tau''<1$ depend on $\eps$ and where we have used the estimates on the classical trajectories. 
Using \eqref{i} and \eqref{ii}, we can write \eqref{phase}.
\\
For \eqref{phase2}, we notice that only one part of the phase depends on $y$, and so 
the $\sqrt{\eps}-$derivative of the phase gives 
$$-\sqrt{\eps} \;i c_0 \; \eps^{-\gamma} \dot{\xi}(-s^\eps\sqrt{\eps}\tau') = \mathcal{O}(\eps^{\frac{1}{2}-\gamma}).$$
\end{proof}
\vspace{2mm}
\noindent
\textbf{Step $(3)$: Analysis of $W^\eps$.} 
\\
$(i)$ We first study $\omega^\eps(-c_0\eps^{-\gamma},y)$: 
it will allow us to see in which weighted space we need to estimate the derivative.
\\
We recall
$$\omega^\eps(s,y)=\eps^{1/4}w^\eps\left(s\sqrt{\eps},\sqrt{\eps}(y+s\xi_0) \right)
e^{\frac{-i\xi_0^2 s}{\sqrt{\eps}}-\frac{i\xi_0.y}{\sqrt{\eps}}},$$
where $w^\eps$ is defined in Section \ref{subsec:profile}, by \eqref{w}, for $s\sqrt{\eps} \in [-T,-c_0\eps^{1/2-\gamma}]$. 
Thanks to Theorem \ref{thm:IncomingOuter}, we deduce
\begin{itemize}
 \item $\|\omega^\eps(-s^\eps)\|_{L^2}=\|w^\eps(-s^\eps)\|_{L^2}\lesssim \eps^{\gamma}$,
 \item Since $\sqrt{\eps}\d_y \omega^\eps = \eps^{1/4}
 \left[\eps\d_x w^\eps+i\xi_0 w^\eps \right]e^{\frac{-i\xi_0^2 s}{\sqrt{\eps}}-\frac{i\xi_0.y}{\sqrt{\eps}}}$ we infer
 $$\|\sqrt{\eps}\d_y\omega^\eps(-s^\eps)\|_{L^2}\lesssim \|\eps \d_x w^\eps(-s^\eps)\|_{L^2}
 +\|w^\eps(-s^\eps)\|_{L^2}\lesssim \eps^{\gamma}.$$
\end{itemize}
$(ii)$ We then study the second term:
\begin{itemize}
\item We first notice that for the $L^2-$norm of this term, we have 
\\
$\displaystyle \left(\int_\R |\theta(\eps)|^2 \left|\int_{0}^{1} \d_y u_\delta (-t^\eps,y+\zeta \theta(\eps)) \; d\zeta \right|^2 dy\right)$
\begin{align*}
 &  \leq |\theta(\eps)|^2 \left(\int_\R \int_0^1 |\d_y u_\delta (-t^\eps,y+\zeta \theta(\eps)) |^2 d\zeta \; dy \right)\\
 &\leq |\theta(\eps)|^2 \|\d_y u_\delta (-t^\eps) \|_{L^2}^2
\end{align*}
 and using \eqref{theta} and Theorem \ref{thm:profile}, we deduce that the second term of $W^\eps$ 
 is an $\mathcal{O}(\eps^{1/2-2\gamma})$ in $L^2$.
 \vspace{2mm}
\item Then for the derivative, we need to study the following terms
\begin{align*}
&\theta(\eps)\sqrt{\eps}\d_y \left(\int_{0}^{1}\d_y u_\delta (-c_0\eps^{1/2-\gamma},\zeta \theta(\eps)+y)\;d\zeta \right)
e^{\frac{i}{\eps}\phi^\eps(y)} \; 
\chi_\delta(\sqrt{\eps}(y-c_0\eps^{-\gamma}))\\
& \quad  + \theta(\eps)\left(\int_{0}^{1}\d_y u_\delta (-c_0\eps^{1/2-\gamma},\zeta \theta(\eps)+y)\;d\zeta \right)e^{\frac{i}{\eps}\phi^\eps(y)}
\eps \; d\chi_\delta(\ldots)\\
& \quad  + \theta(\eps)\left(\int_{0}^{1}\d_y u_\delta (-c_0\eps^{1/2-\gamma},\zeta \theta(\eps)+y)\;d\zeta \right)e^{\frac{i}{\eps}\phi^\eps(y)}
\dfrac{i}{\sqrt{\eps}}\d_y\phi^\eps(y)\chi_\delta(\ldots).
\end{align*}
Using similar arguments as before, we find that the first term presents a  $\mathcal{O}(\eps^{1-2\gamma})$ contribution in $L^2$.
\\
We then use \eqref{eigenvectors} for the second term and obtain that it is a $\mathcal{O}(\eps^{1-2\gamma})$ contribution too.
\\
We finally use \eqref{phase2} for the third term and deduce that the $\sqrt{\eps}-$derivative of the second term of $W^\eps$ 
presents a $\mathcal{O}(\eps^{1-3\gamma})$ contribution in $L^2$.
\end{itemize}
\vspace{2mm}
$(iii)$ For the last term, we write
\begin{multline*}
 \left\|u_\delta(-c_0\eps^{1/2-\gamma},y)e^{\frac{i}{\eps}\phi^\eps(y)}
\left[\chi_\delta(\sqrt{\eps}(y-\xi_0c_0\eps^{-\gamma}))-\left(\begin{array}{c}
 0\\1
\end{array}\right)\right] \right\|_{L^2}^2 =\\
\int_{|y|\leq \theta c_0 \xi_0 \eps^{-\gamma}} \ldots \quad + \int_{|y|\geq \theta c_0 \xi_0 \eps^{-\gamma}} \ldots\quad ,
\end{multline*}
where $\theta <1$ is independent of $\eps$. For the first integral, 
we first recall that in the new variables the eigenvector $\chi_\delta(x)$, associated with $\lambda_\delta(x) =\sqrt{x^2+\delta^2}$ 
is given by (Section \ref{eigen}):
$$\chi_\delta(\sqrt{\eps}(y+s\xi_0))=\left(\begin{array}{c}
                                              -\sin \left(\dfrac{1}{2}{\arctan\left(\dfrac{c}{y+s\xi_0}\right)}\right)\\
                                               \cos \left(\dfrac{1}{2}{\arctan\left(\dfrac{c}{y+s\xi_0}\right)} \right)
                                             \end{array}
 \right) ,$$
for $s \in [-T/\sqrt{\eps}, -c_0\eps^{-\gamma}]$. For $s=-c_0\eps^{-\gamma}$, we compute
$$\chi_\delta(\sqrt{\eps}(y-c_0\xi_0\eps^{-\gamma}))-\left(\begin{array}{c}
                                                       0\\1
                                                       \end{array}\right) = 
                                  \left(\begin{array}{c}
                                 -\sin \left(\dfrac{1}{2}{\arctan\left(\dfrac{c}{y-c_0\xi_0\eps^{-\gamma}}\right)}\right)\\
                                  \cos \left(\dfrac{1}{2}{\arctan\left(\dfrac{c}{y-c_0\xi_0\eps^{-\gamma}}\right)} \right)-1
                                  \end{array}\right),$$
and since $|y|\leq \theta c_0 \xi_0 \eps^{-\gamma}$, $\theta <1$, we have
$$|y- c_0 \xi_0 \eps^{-\gamma} | \geq c_0 \xi_0 \eps^{-\gamma}-|y| \geq (1-\theta)c_0 \xi_0 \eps^{-\gamma},$$
so $$\dfrac{c}{y-c_0\xi_0\eps^{-\gamma}} \leq \dfrac{c}{(1-\theta)c_0\xi_0}\eps^{\gamma}\lesssim \eps^{\gamma}.$$
Then, using Taylor expansions, we have
$$\arctan\left(\dfrac{c}{y-c_0\xi_0\eps^{-\gamma}}\right) = \dfrac{c}{y-c_0\xi_0\eps^{-\gamma}}+ 
\mathcal{O}\left(\left(\dfrac{c}{y-c_0\xi_0\eps^{-\gamma}}\right)^3\right), $$
which gives
$$\sin \left(\dfrac{1}{2}{\arctan\left(\dfrac{c}{y-c_0\xi_0\eps^{-\gamma}}\right)}\right) = 
\dfrac{c}{y-c_0\xi_0\eps^{-\gamma}}+ 
\mathcal{O}\left(\left(\dfrac{c}{y-c_0\xi_0\eps^{-\gamma}}\right)^3\right) = \mathcal{O}(\eps^\gamma),
$$
and
$$\cos \left(\dfrac{1}{2}{\arctan\left(\dfrac{c}{y-c_0\xi_0\eps^{-\gamma}}\right)} \right) = 
1 + \mathcal{O}\left(\left(\dfrac{c}{y-c_0\xi_0\eps^{-\gamma}}\right)^2\right)= 1+ \mathcal{O}(\eps^{2\gamma}).$$
Finally, for $|y|\leq \theta c_0 \xi_0 \eps^{-\gamma}$, we can write
$$\left|\chi_\delta(\sqrt{\eps}(y-c_0\xi_0\eps^{-\gamma}))-\left(\begin{array}{c}
                                                       0\\1
                                                       \end{array}\right)\right| = \mathcal{O}(\eps^{\gamma}),$$
and so, thanks to Theorem \ref{thm:profile}
\begin{equation*}
 \left(\int_{|y|\leq \theta c_0 \xi_0 \eps^{-\gamma}}
 \left|u_\delta(-c_0\eps^{1/2-\gamma},y)\right|^2
\left|\left(\chi_\delta(\sqrt{\eps}(y-c_0\eps^{-\gamma}))-\left(\begin{array}{c}
 0\\1
\end{array}\right)\right)\right|^2 \right)^{1/2} \lesssim \eps^{\gamma}
\end{equation*}
For the second integral, we notice that, since $|y|\geq \theta c_0 \xi_0 \eps^{-\gamma}$, we have
$$\left|\dfrac{y}{c_0\xi_0\eps^{-\gamma}}\right|\geq \theta, $$ and so
$$\left(\int_{|y|\geq \theta c_0 \xi_0 \eps^{-\gamma}}
 \left|u_\delta(-c_0\eps^{1/2-\gamma},y)\right|^2 \times \left|\dfrac{y}{c_0\xi_0\eps^{-\gamma}}\right|^2 \times
 \left|\dfrac{c_0\xi_0\eps^{-\gamma}}{y}\right|^2
 \right)^{1/2}$$
\begin{align*}
 & \leq  
 \dfrac{1}{\theta}
  \left(\int_{|y|\geq \theta c_0 \xi_0 \eps^{-\gamma}}
 \left|u_\delta(-c_0\eps^{1/2-\gamma},y)\right|^2 \times \left|\dfrac{y}{c_0\xi_0\eps^{-\gamma}}\right|^2
 \right)^{1/2} \\
 &\lesssim \eps^{\gamma} \|yu_\delta(-c_0\eps^{1/2-\gamma})\|_{L^2}\lesssim \eps^{\gamma}.
\end{align*}
We then deduce
$$u_\delta(-c_0\eps^{1/2-\gamma},y)e^{\frac{i}{\eps}\phi^\eps(y)}
\left[\chi_\delta(\sqrt{\eps}(y-c_0\eps^{-\gamma}))-\left(\begin{array}{c}
 0\\1
\end{array}\right)\right]= \mathcal{O}(\eps^{\gamma}), \textrm{in }L^2$$ 
and arguing as before, using \eqref{phase2} we can also deduce
$$\sqrt{\eps}\d_y\left[u_\delta(-c_0\eps^{1/2-\gamma},y)e^{\frac{i}{\eps}\phi^\eps(y)}
\left[\chi_\delta(\sqrt{\eps}(y-c_0\eps^{-\gamma}))-\left(\begin{array}{c}
 0\\1
\end{array}\right)\right]\right] = \mathcal{O}(\eps^{\frac{1}{2}-\gamma}), \textrm{in }L^2,$$
and the proof is complete.
\end{proof}

\subsection{Study of the linear approximation $f$}
\hspace{1mm}In this subsection, we prove Corollary \ref{cor:f} and Lemma \ref{lemma:f} concerning $f$ solution to 
\eqref{f}
\begin{equation*}
 i\d_s f - \left(  \begin{array}{cc}
            y+s\xi_0 & c \\
            c & -(y+s\xi_0)
           \end{array}
\right)f = 0 ;
\end{equation*}
with data \eqref{dataf} 
\begin{equation*}
  f(-c_0\eps^{-\gamma},y) = u_\delta(-c_0\eps^{1/2-\gamma},y)\left(\begin{array}{c}
                                                                  0 \\1
                                                                 \end{array}
 \right)e^{\frac{i}{\eps}\phi^\eps(y)}
\end{equation*}
where $s,y\in \R$, $\gamma \in ]0,1/6[$, and $\phi^\eps$ is given in \eqref{phasephi}.
We first prove Corollary~\ref{cor:f}.
\begin{proof}[Proof of Corollary \ref{cor:f}]
We first make a change of variables and switch the coordinates to obtain a similar system: $ \widetilde{s} = (s\xi_0+y)/\sqrt{\xi_0}$ and 
$\widetilde{f}(\widetilde{s},y)= f(s,y)$. The system in the new variables is
\begin{equation*}
 -i\d_{\widetilde{s}} \widetilde{f} - \left(  \begin{array}{cc}
 
            -\widetilde{s} & \dfrac{-c}{\sqrt{\xi_0}} \\
             \dfrac{-c}{\sqrt{\xi_0}}& \widetilde{s}
           \end{array}
\right)\widetilde{f} = 0;
\end{equation*}
and taking $\left\lbrace \begin{array}{cl}
               \widetilde{f}_1  = & u_2 \\
               \widetilde{f}_2  = & u_1
              \end{array},
 \right.$
we obtain
\begin{equation*}
 -i\d_{\widetilde{s}} u - \left(  \begin{array}{cc}
            \widetilde{s} & \dfrac{-c}{\sqrt{\xi_0}} \\
             \dfrac{-c}{\sqrt{\xi_0}}& -\widetilde{s}
           \end{array}
\right)u = 0, \quad \textrm{with }u = (u_1,u_2).
\end{equation*}
We can apply Theorem \ref{diffuFG} to $u$, with $\eta = -c/\sqrt{\xi_0}$ and write 
\begin{align*}
 u(\widetilde{s},y) &= \alpha_1(y)g_1^{s-}(\widetilde{s},y)+\alpha_2(y)g_2^{s-}(\widetilde{s},y)\\
 &= \beta_1(y)g_1^{s+}(\widetilde{s},y)+\beta_2(y)g_2^{s+}(\widetilde{s},y),
 \end{align*}
where the $\alpha_j$ are the coordinates in $(g_1^{s-}(\widetilde{s},y),g_2^{s-}(\widetilde{s},y)) $ and
the $\beta_j$ are the coordinates in $(g_1^{s+}(\widetilde{s},y),g_2^{s+}(\widetilde{s},y)).$ 
Using asymptotics of the bases from \cite{FG02}, and since 
$$f_{1,2}(s,y)=\widetilde{f}_{1,2}(\widetilde{s},y)=u_{2,1}(\widetilde{s},y),$$ we can deduce
\begin{align*}
 f(s,y) &= \alpha_2(y) e^{-i\Lambda \left(\frac{s\xi_0+y}{\sqrt{\xi_0}},y\right)}\left(\begin{array}{c}
                                                                                1\\0
                                                                               \end{array}\right)
+\alpha_1(y) e^{i\Lambda \left(\frac{s\xi_0+y}{\sqrt{\xi_0}},y\right)}\left(\begin{array}{c}
                                                                                0\\1
                                                                               \end{array}\right)\\
& \qquad + \alpha_1(y)\rho^{s-}_2\left(\dfrac{\xi_0 s+y}{\sqrt{\xi_0}}\right) 
+\alpha_2(y)\rho^{s-}_1\left(\dfrac{\xi_0 s+y}{\sqrt{\xi_0}}\right)\\
& = \beta_2(y) e^{-i\Lambda \left(\frac{s\xi_0+y}{\sqrt{\xi_0}},y\right)}\left(\begin{array}{c}
                                                                                1\\0
                                                                               \end{array}\right)
+\beta_1(y) e^{i\Lambda \left(\frac{s\xi_0+y}{\sqrt{\xi_0}},y\right)}\left(\begin{array}{c}
                                                                                0\\1
                                                                               \end{array}\right)\\
&\qquad + \beta_1(y)\rho^{s-}_2\left(\dfrac{\xi_0 s+y}{\sqrt{\xi_0}}\right)
+\beta_2(y)\rho^{s-}_1\left(\dfrac{\xi_0 s+y}{\sqrt{\xi_0}}\right),
\end{align*}
where $$\Lambda\left(\frac{s\xi_0+y}{\sqrt{\xi_0}},y\right) =  
\dfrac{\left(\frac{s\xi_0+y}{\sqrt{\xi_0}}\right)^2}{2}+\dfrac{\eta^2}{2}\log\left| \frac{s\xi_0+y}{\sqrt{\xi_0}}\right|,
\quad \rho^{s\pm}_{1,2}(s) = \mathcal{O}\left( \dfrac{1}{s}\right) \; \textrm{in }L^2.$$
Expanding $\Lambda\left(\frac{-s^\eps\xi_0+y}{\sqrt{\xi_0}},y\right)$, 
and comparing with \eqref{dataf}, we infer
$$\alpha_2 \equiv 0, \quad \textrm{and }\; \alpha_1(y) = u_\delta(-c_0\eps^{1/2-\gamma},y) e^{i\theta^\eps(y)}$$
 \begin{align*}
 i\theta^\eps(y) &= ic\;c_0\eps^{-\gamma}+\dfrac{ic_0^2\eps^{-2\gamma}}{2}
 \left(-\xi_0+\dot{\xi}(-c_0\eps^{1/2-\gamma}\tau_\eps)(\xi_0+\dot{\xi}(-c_0\eps^{1/2-\gamma}\tau_\eps)) \right)
 \\
 &\quad -\dfrac{ic_0^3S^{(3)}(-c_0\eps^{1/2-\gamma}{\tau'}_\eps)}{6}\eps^{1/2-3\gamma}
 -ic_0\eps^{-\gamma}\left(\dot{\xi}(-c_0\eps^{1/2-\gamma}\tau_\eps)-1\right)y \\
 & \quad - \dfrac{iy^2}{\xi_0}
 +\dfrac{ic^2}{2\xi_0}\log\left|\frac{-c_0\eps^{-\gamma}\xi_0+y}{\sqrt{\xi_0}} \right|
\end{align*}
Then, using 
the relation between the coordinates in each basis, we find
\begin{align*}
\beta_1(y)& =a(\eta)\alpha_1(y)-\overline{b}(\eta)\alpha_2(y) 
= a(\eta) u_\delta(-c_0\eps^{1/2-\gamma},y) e^{i\theta^\eps(y)}\\
\beta_2(y)& = b(\eta) \alpha_1(y) +a(\eta)\alpha_2(y) = b(\eta) u_\delta(-c_0\eps^{1/2-\gamma},y) e^{i\theta^\eps(y)}.
\end{align*}
Besides, using
$$\rho^{s-}_2\left(\dfrac{\xi_0 s+y}{\sqrt{\xi_0}}\right) = g_2^{s-}\left(\dfrac{\xi_0 s+y}{\sqrt{\xi_0}}\right)-
e^{-i\Lambda\left(\dfrac{\xi_0 s+y}{\sqrt{\xi_0}},y\right)}\left(\begin{array}{c}
                                                                  0 \\1
                                                                 \end{array}
 \right), $$
 and the fact that $(g_1^{s-},g_2^{s-})$ is an orthonormal basis, we have $|\rho^{s-}|_{\C^2} \leq 2$.
Moreover, for all $y\in \R$, $\rho^{s-}$ tends to zero when $\eps$ tends to zero. 
 Then, thanks to mass-conservation of $u_\delta$ and using Lebesgue's Dominated Convergence Theorem, one can deduce that 
 $$u_\delta(-c_0\eps^{1/2-\gamma}\xi_0, y)\; \rho^{s-}\left(\dfrac{\xi_0 s+y}{\sqrt{\xi_0}}\right) = o(1) \; \textrm{in }L^2.$$
 The same argument holds for other terms containg $\rho^{s\pm}$.
So, the $L^2$-norms for each coordinates, at time $-s^\eps$ and $s^\eps$ can be computed, using that 
$\|\alpha_2\|_{L^2}=\|a\|_{L^2}$:
\begin{align*}
& \|f_1(-c_0\eps^{-\gamma}) \|_{L^2}^2=0, \hspace{3cm}\|f_2(-c_0\eps^{-\gamma}) \|_{L^2}^2= \|a\|_{L^2}^2+ o(1),\\
&\|f_1(c_0\eps^{-\gamma}) \|_{L^2}^2= |b(\eta)|^2\|a\|_{L^2}^2 + o(1), \hspace{5mm} \|f_2(c_0\eps^{-\gamma}) \|_{L^2}^2= |a(\eta)|^2\|a\|_{L^2}^2 + o(1),
\end{align*}
and the proof is complete.
\end{proof}
Since the function $f$ depends on $\eps$ because of its data, we have to study its derivatives 
in order to understand if it implies a loss of power of $\eps$.
\begin{proof}[Proof of Lemma \ref{lemma:f}]
We proceed by induction. Using \eqref{f}, it is easily seen that
\begin{equation*}
 \dfrac{d}{ds}\|f\|_{L^2}^2 = 2\; Im \left\langle{f}|i\d_s f \right\rangle =0.
\end{equation*}
So for all $s \in [-c_0\eps^{-\gamma},c_0\eps^{-\gamma}]$ we have 
$$\|f(s)\|_{L^2}=\|f(-c_0\eps^{-\gamma})\|_{L^2}=\|u_\delta(-c_0\eps^{1/2-\gamma})\|_{L^2}=C_0.$$
The following step poses no problem, since $\d_y f$ satisfies:
$$i\d_s (\d_y f) - \begin{pmatrix}
                            y+s\xi_0 & c \\
                            c & -(y+s\xi_0)
                           \end{pmatrix}\d_y f 
                           = \begin{pmatrix}
                           1 & 0 \\
                           0 & -1
                           \end{pmatrix}f,$$
and
\begin{align*}
 \dfrac{d}{ds}\|\d_y f\|_{L^2}^2&=2\|\d_y f\|_{L^2}.\dfrac{d}{ds}\|\d_yf\|_{L^2}\\
 & = 2\; Im \left\langle{\d_yf}|i\d_s \d_y f \right\rangle \\
 &\leq C \| \d_yf\|_{L^2},
\end{align*}
and so
\begin{equation*}
\|\d_y f(s)\|_{L^2} \leq C|s|+ \|\d_y f(-c_0\eps^{-\gamma})\|_{L^2},
\end{equation*}
for $|s|\leq c_0 \eps^{-\gamma}$ and where 
\begin{align*}
 \d_y f(-c_0\eps^{-\gamma})&=\d_y \left(u_\delta (s,y)e^{\frac{i}{\eps}\phi^\eps(y)}\right) 
 \\
 & = 
\left(\d_y u_\delta (s,y)+\dfrac{d}{dy}\left(\dfrac{i}{\eps}\phi^\eps(y)\right)u_\delta (s,y)\right)e^{\frac{i}{\eps}\phi^\eps(y)}.
\end{align*}
Using \eqref{phase2} we deduce
\begin{equation*}
\|\d_y f(s)\|_{L^2} \leq C_1 \eps^{-\gamma}.
\end{equation*}
\\[1mm]
Now for $k\geq 1$, an easy computation shows that $\d_y^k f$ satisfies
$$i\d_s (\d_y^k f) - \begin{pmatrix}
                            y+s\xi_0 & c \\
                            c & -(y+s\xi_0)
                           \end{pmatrix}\d_y^k f 
                           = k \begin{pmatrix}
                           1 & 0 \\
                           0 & -1
                           \end{pmatrix}\d_y^{k-1}f,$$
and that
\begin{align*}
 \dfrac{d}{ds}\|\d_y^k f\|_{L^2}^2&=2\|\d_y^k f\|_{L^2}.\dfrac{d}{ds}\|\d_y^k f\|_{L^2}\\
 & = 2\; Im \left\langle{\d_y^kf}|i\d_s \d_y^k f \right\rangle \\
 &\leq \widetilde{C}_k \| \d_y^k f\|_{L^2}\| \d_y^{k-1} f\|_{L^2} \\
 &\leq \widetilde{C}_k \| \d_y^k f\|_{L^2} \eps^{-(k-1)\gamma},
\end{align*}
where we have used the induction hypothesis. Then, since $\phi^\eps(y)$ is linear in $y$ (see \eqref{phasephi}) 
we then easily notice that for all $k \in \N$
$$\d_y^k \left(u_\delta (s,y)e^{\frac{i}{\eps}\phi^\eps(s,y)}\right) = 
\sum_{j=0}^{k}C_j \d_y^ju_\delta(s,y)\; e^{\frac{i}{\eps}\phi^\eps}\left(\dfrac{d}{dy}\dfrac{i}{\eps}\phi^\eps(y) \right)^{k-j}.$$
So using \eqref{phase2}, we deduce for $|s|\lesssim \eps^{-\gamma}$
$$\|\d_y^k f(s)\|_{L^2} \leq \|\d_y^k f(-c_0\eps^{-\gamma})\|_{L^2} +\widetilde{C}_k |s| \eps^{-(k-1)\gamma} 
\leq C_k \eps^{-k\gamma},$$
which completes the proof.                            
\end{proof}
The last step is to check the validity of the approximation on the interval we consider.

\subsection{Validity of the inner approximation and conclusion}
We now prove Theorem \ref{thm:Inner}.
\\[1mm]
We use 
$$i\d_s r^\eps + \dfrac{\sqrt{\eps}}{2}\d_y^2 r^\eps - \begin{pmatrix}
                                                        y+s\xi_0 & c \\
                                                        c & -(y+s\xi_0)
                                                       \end{pmatrix}r^\eps 
                                                       = \dfrac{\sqrt{\eps}}{2}\d_y^2f+\kappa\sqrt{\eps}|r^\eps+f|^2(r^\eps+f)$$
to compute
\begin{align*}
 \dfrac{d}{ds} \|r^\eps(s) \|_{L^2}^2 &= 2 \|r^\eps(s) \|_{L^2}. \dfrac{d}{ds}\|r^\eps(s) \|_{L^2}\\
 &=2 \; Im \left\langle \overline{r^\eps} \left| \dfrac{-\sqrt{\eps}}{2}\d_y^2 r^\eps + V_\delta(s\xi_0+y)r^\eps \right.\right. \\
 & \qquad \left.+ \dfrac{\sqrt{\eps}}{2}\d_y^2 f + \kappa\sqrt{\eps}|r^\eps+f|^2 (r^\eps+f) \right\rangle.
\end{align*}
We check that the first and second terms give no contribution since they are real-valued when 
one computes the scalar products with $r^\eps$. For the remaining terms, we observe
\begin{align*}
 2 \; Im \left\langle \overline{r^\eps} | \dfrac{\sqrt{\eps}}{2} \d_y^2 f \right\rangle &\leq \sqrt{\eps}\|r^\eps(s)\|_{L^2} \; \|\d_y^2f\|_{L^2}\\
 2 \; Im \left\langle \overline{r^\eps} | \sqrt{\eps} |r^\eps+f|^2(r^\eps+f) \right\rangle  &\lesssim \kappa \sqrt{\eps}\left(\|r^\eps\|_{L\infty}^2
 +\|f\|_{L^\infty}^2\right)\|f\|_{L^2}\|r^\eps\|_{L^2}.
\end{align*}
We perform a bootstrap argument with the following assumption
\begin{equation}
 \label{bootstrap2}
 \|r^\eps \|_{L^\infty}^2\leq \eps^{-\alpha},
\end{equation}
with $\alpha = (1-3\gamma)/2$.
Besides, we use Gagliardo-Nirenberg inequality and Lemma \ref{lemma:f} to write
$$\|f\|_{L^\infty}\lesssim \|f\|_{L^2}^{1/2}\|\d_y f\|_{L^2}^{1/2} \lesssim \eps^{-\gamma/2}.$$
These points allow us to write
\begin{align*}
\dfrac{d}{ds} \|r^\eps(s)\|_{L^2} &\lesssim \sqrt{\eps}\|\d_y^2f\|_{L^2} + \kappa\sqrt{\eps}\left(\|r^\eps\|_{L\infty}^2
 +\|f\|_{L^\infty}^2\right)\|f\|_{L^2} \\
 &\leq C \left[ \eps^{1/2-2\gamma}+\kappa \eps^{1/2-\alpha}+ \kappa \eps^{1/2-\gamma}\right]\\
 & \leq C \left[ \eps^{1/2-2\gamma}+\kappa \eps^{3\gamma/2}\right].
\end{align*}
Integrating these terms, we obtain for $s \in [-c_0 \eps^{-\gamma},c_0 \eps^{-\gamma}]$
\begin{align}
\notag
 \|r^\eps(s)\|_{L^2} &\leq \|r^\eps(-c_0 \eps^{-\gamma})\|_{L^2} + \widetilde{C}\left[
 \eps^{1/2-3\gamma} + \kappa \eps^{\gamma/2}\right] \\
 \notag
 & \leq 
 C_1\eps^{\gamma} +
C \eps^{1/2-3\gamma} + C\kappa \eps^{\gamma/2}\\
 \label{r}
 &\leq C\kappa \eps^{\gamma/2},
\end{align}
where $C_1$ is given by \eqref{wL2}, and where we only keep the worst contributions. 
\\[1mm]
We now wish to prove the validity of the bootstrap assumption \eqref{bootstrap2}. 
Thus, we look for the weighted derivative of $r^\eps$, which satisfies
\begin{multline*}i\d_s (\sqrt{\eps}\d_yr^\eps) + \dfrac{\sqrt{\eps}}{2}\d_y^2 (\sqrt{\eps}\d_yr^\eps) - \begin{pmatrix}
                                                        y+s\xi_0 & c \\
                                                        c & -(y+s\xi_0)
                                                       \end{pmatrix}(\sqrt{\eps}\d_yr^\eps) = \\
                                                       \begin{pmatrix}
                                                        1 & 0\\
                                                        0 & -1
                                                       \end{pmatrix}\sqrt{\eps} r^\eps  +
                                             \dfrac{\eps}{2}\d_y^3f+\kappa\;\eps\d_y\left(|r^\eps+f|^2(r^\eps+f)\right). 
                                             \end{multline*}
We first notice that
\begin{equation*}
 \d_y\left(|r^\eps+f|^2(r^\eps+f)\right) = (r^\eps+f)^2 \left(\overline{\d_y r^\eps}+\overline{\d_y f} \right)
 + 2 |r^\eps+f|^2 \left(\d_y r^\eps +\d_y f \right).
\end{equation*}
Proceeding in the same way as before, and thanks to \eqref{bootstrap2}, \eqref{r} and Lemma \ref{lemma:f}, we see that
\begin{align*}
 \dfrac{d}{ds}\|\sqrt{\eps}\d_y r^\eps\|_{L^2} & \lesssim \sqrt{\eps}\|r^\eps\|_{L^2}+ \eps \|\d_y^3 f\|_{L^2} 
 \\
 & \quad + 
 \kappa\sqrt{\eps}\left[\|r^\eps \|_{L^\infty}^2 + \|f \|_{L^\infty}^2 \right]\; 
 \left[\|\sqrt{\eps}\d_y r^\eps\|_{L^2}+\|\sqrt{\eps}\d_y f \|_{L^2} \right]\\
 \\
 & \lesssim \kappa \eps^{1/2+\gamma/2}+\eps^{1-3\gamma}+\kappa\eps^{1-\alpha-\gamma}+\kappa\eps^{1-2\gamma}\\
 & \quad + \left(\kappa \eps^{1/2-\alpha}+\kappa\eps^{1/2-\gamma} \right)\|\sqrt{\eps}\d_y r^\eps\|_{L^2} \\
 \\
 & \lesssim \kappa \eps^{1/2+\gamma/2}+\eps^{1-3\gamma} 
 + \left(\kappa \eps^{3\gamma/2}+\kappa\eps^{1/2-\gamma} \right)\|\sqrt{\eps}\d_y r^\eps\|_{L^2}.
\end{align*}
We then find by integration, for $|s| \leq c_0 \eps^{-\gamma}$
\begin{align*}
 \|\sqrt{\eps}\d_y r^\eps(s)\|_{L^2} &\leq  \|\d_y r^\eps(-c_0 \eps^{-\gamma})\|_{L^2}+ C \left[
  \eps^{1-4\gamma}+\kappa \eps^{1/2-\gamma/2}\right.\\
 &\quad \left. +\int_{-c_0\eps^{-\gamma}}^{s} 
 \left(\kappa \eps^{3\gamma/2}+\eps^{1/2-\gamma} \right)\|\sqrt{\eps}\d_y r^\eps(z)\|_{L^2}dz\right],
\end{align*}
and using Gronwall Lemma, we deduce for $0<\gamma<1/6$
\begin{align}
\notag
  \|\sqrt{\eps}\d_y r^\eps(s)\|_{L^2} & \leq C_2 \eps^\gamma +  C\left[\eps^{1-4\gamma}+\kappa \eps^{1/2-\gamma/2} \right]
  e^{C(\eps^{1/2-2\gamma}+\kappa\eps^{\gamma/2})}\\
\label{dr}
  & \leq C_2 \eps^{\gamma},
\end{align}
where $C_2$ is given by \eqref{wH1}, where we have only kept the worst term 
and since $C\eps^{1/2-\gamma}+\kappa\eps^{3\gamma/2}\leq 1$, for $\eps$ sufficiently small, 
the exponential term is bounded independently of $\eps$ on the time interval we consider.
\\[1mm]
Then, Gagliardo Nirenberg inequality gives us
\begin{align*}
 \|r^\eps\|_{L^\infty}^2 &\leq C_{GN}^2 \| r^\eps(s)\|_{L^2}\;\|\d_y r^\eps(s)\|_{L^2}\\
 & \leq \widetilde{C}\kappa \eps^{\gamma/2}\eps^{\gamma-1/2} \leq \widetilde{C}\kappa \eps^{\frac{3\gamma-1}{2}},
\end{align*}
and the bootstrap assumption \eqref{bootstrap2} holds as long as 
$$\widetilde{C}\kappa \eps^{\frac{3\gamma-1}{2}} \leq \eps^{\frac{3\gamma-1}{2}} \Leftrightarrow \widetilde{C}\kappa  \leq 1,$$
and for a small coefficient $\kappa < 1/\widetilde{C}$, the inequality is true and the proof of Theorem~\ref{thm:Inner} is complete.
\section{Transition between the modes}
\label{conclusion}
\noindent
In this section, we prove Corollary \ref{cor:main} to conclude.
\\
By definition of $v^\eps$, we have
$$\psi^\eps(t^\eps,x)  = \eps^{-1/4}v^\eps\left(\dfrac{t^\eps}{\sqrt{\eps}},\dfrac{x-t^\eps\xi_0}{\sqrt{\eps}} \right)
e^{i\frac{\xi_0^2t^\eps}{2\eps}+i\frac{\xi_0(x-t^\eps\xi_0)}{\eps}}$$
and from Theorem \ref{thm:Inner} we infer the following estimate in $L^2$:
\begin{equation*}
 \left( \begin{array}{c}
\psi^\eps_1(t^\eps,x) \\
\psi^\eps_2(t^\eps,x)
\end{array}\right) = \eps^{-1/4}\left( \begin{array}{c}
                             f_1 \left(\dfrac{t^\eps}{\sqrt{\eps}},\dfrac{x-t^\eps\xi_0}{\sqrt{\eps}} \right)\\
                             f_2 \left(\dfrac{t^\eps}{\sqrt{\eps}},\dfrac{x-t^\eps\xi_0}{\sqrt{\eps}} \right)
                          \end{array}\right)
e^{i\frac{\xi_0^2t^\eps}{2\eps}+i\frac{\xi_0(x-t^\eps\xi_0)}{\eps}}
+ \mathcal{O}(\eps^{\gamma/2}),
\end{equation*}
where the last term contains the $L^2-$ norm of $r^\eps$.
We then write
$$\psi^\eps_\pm(t^\eps,x)= \left\langle \psi^\eps(t^\eps,x)|\chi_\delta^\pm(x) \right\rangle, $$
and using Section \ref{eigen}:
\begin{align*}
 \psi^\eps_+ (t^\eps,x)& = \psi^\eps_1 (t^\eps,x)\cos \left(\Phi\left(\dfrac{x}{\delta}\right)\right)+ 
 \psi^\eps_2 (t^\eps,x) \sin \left(\Phi\left(\dfrac{x}{\delta}\right)\right)\\
 \psi^\eps_- (t^\eps,x)& = \psi^\eps_1 (t^\eps,x) \sin \left(\Phi\left(\dfrac{x}{\delta}\right)\right)- 
 \psi^\eps_2 (t^\eps,x) \cos \left(\Phi\left(\dfrac{x}{\delta}\right)\right),
\end{align*}
where $\Phi$ is given by 
$\Phi\left(\dfrac{x}{\delta}\right) = \dfrac{1}{2}\arctan\left(\dfrac{\delta}{x}\right).$
\\
For $\psi^\eps_+$, we have
\begin{align*}
 \|\psi^\eps_+(t^\eps) \|_{L^2}^2 & = \int_{y\in \R} \left|f_1\left(\dfrac{t^\eps}{\sqrt{\eps}},y\right) 
 \cos\left(\dfrac{1}{2}\arctan\left(\dfrac{c}{y+\frac{t^\eps}{\sqrt{\eps}}\xi_0}\right)\right) \right. \\
 & \left.\quad + f_2\left(\dfrac{t^\eps}{\sqrt{\eps}},y\right) 
 \sin\left(\dfrac{1}{2}\arctan\left(\dfrac{c}{y+\frac{t^\eps}{\sqrt{\eps}}\xi_0}\right)\right)\right|^2  dy + 
 \mathcal{O}(\eps^{\gamma/2})\\
 \\
  \|\psi^\eps_-(t^\eps) \|_{L^2}^2 & = \int_{y\in \R} \left|f_2\left(\dfrac{t^\eps}{\sqrt{\eps}},y\right) 
 \cos\left(\dfrac{1}{2}\arctan\left(\dfrac{c}{y+\frac{t^\eps}{\sqrt{\eps}}\xi_0}\right)\right) \right.\\
 & \left. \quad -
 f_1\left(\dfrac{t^\eps}{\sqrt{\eps}},y\right) 
 \sin\left(\dfrac{1}{2}\arctan\left(\dfrac{c}{y+\frac{t^\eps}{\sqrt{\eps}}\xi_0}\right)\right)\right|^2  dy + 
 \mathcal{O}(\eps^{\gamma/2}).
\end{align*}
We split up the integrals into:
$$ \int_{y\in \R} \ldots dy = 
 \int_{|y|\geq \frac{\tau\xi_0\sqrt{\eps}}{t^\eps}} \ldots+\int_{|y|\leq \frac{\tau\xi_0\sqrt{\eps}}{t^\eps}}\ldots $$
 for $0<\tau<1$ and we use the same arguments as in Section \ref{match}.
 \begin{itemize}
  \item For $|y|\geq \frac{\tau\xi_0\sqrt{\eps}}{t^\eps}$, we have $|y|\sqrt{\eps}/(\xi_0t^\eps)\geq \tau $ and with 
  Corollary \ref{cor:f} we have
  \begin{align*}
   \int_{|y|\geq \frac{\tau\xi_0\sqrt{\eps}}{t^\eps}} \ldots &\lesssim
   \int_{|y|\geq \frac{\tau\xi_0\sqrt{\eps}}{t^\eps}} |f_1|^2+|f_2|^2 \; dy \\
   & \lesssim \int_{|y|\geq \frac{\tau\xi_0\sqrt{\eps}}{t^\eps}} |u_\delta(-t^\eps,y)|^2 \times 
   \dfrac{\eps|y|^2}{\xi_0^2(t^\eps)^2} \times \dfrac{\xi_0^2(t^\eps)^2}{\eps|y|^2} dy \\
   &= \mathcal{O}\left(\left(\dfrac{t^\eps}{\sqrt{\eps}}\right)^2\right) = \mathcal{O}(\eps^{2\gamma})
  \end{align*}

  \item For $|y|\leq \frac{\tau\xi_0\sqrt{\eps}}{t^\eps}$, we have 
  $$\left|y+\dfrac{\tau t^\eps \xi_0}{\sqrt{\eps}} \right|\geq  \dfrac{\tau t^\eps \xi_0}{\sqrt{\eps}} - |y| 
  \geq (1-\tau)\dfrac{\xi_0 t^\eps}{\sqrt{\eps}}, \quad \textrm{where }\tau \in ]0,1[.$$
 We deduce $\left|y+\dfrac{\tau t^\eps \xi_0}{\sqrt{\eps}} \right|^{-1} \leq \mathcal{O}(\eps^\gamma) $ and so
 we can use the same asymptotics as in Section \ref{match} and write
  $$\arctan\left(\dfrac{c}{y+c_0\xi_0\eps^{-\gamma}}\right) = \dfrac{c}{y+c_0\xi_0\eps^{-\gamma}}+ 
\mathcal{O}\left(\left(\dfrac{c}{y+c_0\xi_0\eps^{-\gamma}}\right)^3\right), $$
which gives
$$\sin \left(\dfrac{1}{2}{\arctan\left(\dfrac{c}{y+c_0\xi_0\eps^{-\gamma}}\right)}\right) = \mathcal{O}(\eps^\gamma),
$$
and
$$\cos \left(\dfrac{1}{2}{\arctan\left(\dfrac{c}{y+c_0\xi_0\eps^{-\gamma}}\right)} \right) = 1+ \mathcal{O}(\eps^{2\gamma}).$$
We then compute for $\psi^\eps_+$:
\begin{align*}
 \int_{|y|\leq \frac{\tau\xi_0\sqrt{\eps}}{t^\eps}}\ldots & = \int_{|y|\leq \frac{\tau\xi_0\sqrt{\eps}}{t^\eps}}
 |b(c/\sqrt{\xi_0})|^2 |u_\delta(-t^\eps,y)|^2 \left(1+\mathcal{O}(\eps^{3\gamma}) \right) dy \\
 & = \left(1+ \mathcal{O}(\eps^{3\gamma}) \right)|b(c/\sqrt{\xi_0})|^2 \|a \|_{L^2}^2,
\end{align*}
and since $\|\psi^\eps_-(t^\eps)\|^2_{L^2}=\|a\|^2_{L^2}-\|\psi^\eps_+(t^\eps)\|^2_{L^2} = 
\left(e^{-\frac{\pi c^2}{\xi_0}}+\mathcal{O}(\eps^{3\gamma}) \right)\|a\|_{L^2}^2$
 \end{itemize}
and Corollary \ref{cor:main} is proved. We can then infer that the propagation through an avoided crossing point of this type 
generates transition of energy between both levels.
\\[5mm]
\textbf{Acknowledgments.} The author whishes to express his gratitude to Thomas Duyckaerts and Clotilde Fermanian for suggesting 
the problem and for many helpful and stimulating discussions.
\bibliographystyle{amsplain}
\bibliography{biblio}
\end{document}